\def\notes{0}
\def\pods{0}

\ifnum\pods=0
\documentclass{article}
\else
\documentclass[sigconf]{acmart}
\fi

\usepackage[utf8]{inputenc}
\usepackage{bbm}

\ifnum\pods=0
\usepackage{amssymb} 
\else
\fi

\ifnum\pods=0
\usepackage[dvipsnames]{xcolor} 
\else
\fi

\usepackage{amsmath,fullpage}
\usepackage{amsthm}
\usepackage{physics}
\usepackage{algorithm,algorithmicx}
\usepackage{subfiles} 
\usepackage[noend]{algpseudocode}
\usepackage{comment}
\usepackage{xspace}
\usepackage[inline,shortlabels]{enumitem}
\usepackage{mathtools}
\usepackage{ulem} 

\usepackage{setspace}
\usepackage{hyperref}
\usepackage{cleveref}

\newtheorem{theorem}{Theorem}[section]

\newtheorem{corollary}[theorem]{Corollary}
\newtheorem{claim}[theorem]{Claim}
\newtheorem{lemma}[theorem]{Lemma}

\newtheorem{definition}[theorem]{Definition}
\newtheorem*{remark}{Remark}

\newcommand{\msf}{\mathsf} 

\newcommand{\multiline}[1]{%
  \begin{tabularx}{\dimexpr\linewidth-\ALG@thistlm}[t]{@{}X@{}}
    #1
  \end{tabularx}
}

\usepackage{stackengine}
\setstackEOL{\\}

\algnewcommand{\LineComment}[1]{\Statex \(\triangleright\) \begin{small}{\emph #1}\end{small} }

\newcommand{\CR}{nonadaptive continual release model} 
\newcommand{\aCR}{adaptive continual release model} 
\renewcommand{\epsilon}{\varepsilon} 
\newcommand{\eps}{\varepsilon} 
\newcommand{\X}{\mathcal{X}} 
\newcommand{\Z}{\mathcal{Z}} 
\newcommand{\BigO}[1]{\ensuremath{\operatorname{O}\left(#1\right)}} 
\newcommand{\dset}{y} 
\newcommand{\dstream}{x} 
\renewcommand{\vec}[1]{{\bf #1}} 
\newcommand{\dstreamvec}{\vec{x}}
\newcommand{\avec}{\vec{a}}
\newcommand{\yvec}{\vec{y}}
\newcommand{\random}{ } 
\newcommand{\xt}{\vec{\dstream}_{[t]}} 
\newcommand{\node}{v}
\newcommand{\nodeval}{s}
\newcommand{\adv}{\ensuremath{\mathcal{A}\mathit{dv}}\xspace} 

\newcommand{\mech}{\ensuremath{\mathcal{M}}\xspace} 
\newcommand{\alg}{\ensuremath{\mathcal{A}}\xspace} 
\newcommand{\Sim}{\mathcal{S}\text{im}} 
\newcommand{\maxsum}[1]{\ensuremath{\mathrm{\sf MaxSum}_{#1}}} 
\newcommand{\kindselection}[1]{\ensuremath{\text{$k$-}\mathrm{\sf Select}_{#1}}} 
\newcommand{\marginals}[1]{\ensuremath{\mathrm{\sf Marginals}_{#1}}} 
\newcommand{\selection}[1]{\ensuremath{\mathrm{\sf SumSelect}_{#1}}} 
\newcommand{\idealgauss}{\mech_{\msf{gauss}}}
\newcommand{\idealrecomp}{\mech_{\msf{comp}}}
\newcommand{\err}{\text{\sf ERR}} 
\newcommand{\side}{\text{\sf side}} 
\newcommand{\gauss}{\ensuremath{\mathcal{N}}} 
\newcommand{\Gauss}{\gauss} 
\newcommand{\Lap}{\mathsf{Lap}}

\newcommand{\zo}{\{0,1\}} 
\newcommand{\R}{\mathbb{R}} 
\newcommand{\N}{\mathbb{N}} 

\newcommand{\defeq}{\overset{\text{\tiny def}}{=}} 
\DeclareMathOperator*{\E}{\mathbb{E}} 
\DeclareMathOperator*{\argmax}{arg\,max} 

\newcommand{\paren}[1]{{\left ( {#1} \right)}}
\newcommand{\bparen}[1]{{\big ( {#1} \big)}}
\newcommand{\Bparen}[1]{{\Big ( {#1} \Big)}}

\definecolor{mygreen}{rgb}{0.0, 0.49, 0.28}

\ifnum\notes=1 
\newcommand{\srnote}[1]{{\color{blue}\footnote{{\color{blue} {\bf SR:} #1}}}}
\newcommand{\pjnote}[1]{{\color{mygreen}\footnote{{\color{mygreen} {\bf P:} #1}}}}
\newcommand{\ssnote}[1]{{\color{red}\footnote{{\color{red} {\bf SS:} #1}}}}
\newcommand{\asnote}[1]{{\color{brown}\footnote{{\color{brown} {\bf A:} #1}}}}
\newcommand{\old}[1]{\sout{#1}}
\else
\newcommand{\srnote}[1]{}
\newcommand{\ssnote}[1]{}
\newcommand{\pjnote}[1]{}
\newcommand{\asnote}[1]{}
\newcommand{\old}[1]{}
\fi 

\ifnum\notes=1

\newcommand{\sstext}[1]{{\color{red} #1}}
\newcommand{\as}[1]{{\color{brown} #1}}
\else 

\newcommand{\sstext}[1]{{#1}}
\newcommand{\as}[1]{{#1}}
\fi

\title{
The Price of Differential Privacy under Continual Observation \\}

\ifnum\pods=0
 \author{Palak Jain\thanks{Department of Computer Science, Boston University. \texttt{\{palakj,satchit,sofya,ads22\}@bu.edu}.
 Palak Jain and Adam Smith were supported in part by NSF award CCF-1763786 as well as a Sloan Foundation research award.
 Sofya Raskhodnikova was partially supported by NSF award CCF-1909612.
 Satchit Sivakumar was supported in part by NSF award CNS-2046425, as well as Cooperative Agreement CB20ADR0160001 with the Census Bureau. The views expressed in this paper are those of the authors and not those of the U.S. Census Bureau or any other sponsor.} \and Sofya Raskhodnikova\footnotemark[1] \and Satchit Sivakumar\footnotemark[1] \and Adam Smith\footnotemark[1]}
 \date{\today}
\else
\fi

\begin{document}
\ifnum\pods=0
\maketitle 
\else
\fi

\begin{abstract}
We study the accuracy of differentially private mechanisms in the continual release model. A continual release mechanism receives a sensitive dataset as a stream of $T$ inputs and produces, after receiving each input, an  accurate output on the obtained inputs.
In contrast, a batch algorithm receives the data as one batch and produces a single output.

We provide the first strong lower bounds on the error of continual release mechanisms. In particular, for two fundamental problems that are widely studied and used in the batch model, we show that the worst case error of every continual release algorithm is $\tilde \Omega(T^{1/3})$ times larger than that of the best batch algorithm. 
Previous work shows only a polylogarithimic (in $T$) gap between the worst case error achievable in these two models; further, for many problems, including the summation of binary attributes, the polylogarithmic gap is tight (Dwork et al., 2010; Chan et al., 2010). Our results show that  problems closely related to summation---specifically, those that require selecting the largest of a set of sums---are fundamentally harder in the continual release model than in the batch model.

Our lower bounds assume only that privacy holds for streams fixed in advance (the ``nonadaptive'' setting). However, we provide  matching upper bounds that hold in a model where privacy is required even for adaptively selected streams. This model may be of independent interest.

\end{abstract}

\ifnum\pods=1
\settopmatter{printacmref=false}
\else
\fi

\maketitle 

\ifnum\pods=0
\newpage  
{ \setstretch{1}
\tableofcontents
}
\newpage
\else
\fi


\section{Introduction}

In fields ranging from healthcare to criminal justice, sensitive data is being analyzed to identify patterns and draw population-level conclusions. Differentially private (DP) data analysis \cite{dwork2006calibrating} studies the design of algorithms that publish such aggregate statistics about input datasets while preserving the privacy of individuals whose data they contain. Differential privacy has been extensively studied and DP algorithms have been deployed in both industry and government.
Current government deployments, notably at the US Census Bureau~\cite{censuscite}, operate in the \textit{batch model}: that is, they collect their input all at once and produce a single output.
However, in many situations, the data are collected over time, and the published statistics need to be updated regularly. An example of such a statistic is the number of COVID-19 cases. To investigate privacy in these situations, Dwork et al.~\cite{DworkNPR10} and Chan et al.~\cite{ChanSS10} introduced the {\em continual release model}.
In this model, a mechanism receives a sensitive dataset as a stream of $T$ input records and produces, after receiving each record, an  accurate output on the obtained inputs. Intuitively, the mechanism is differentially private if releasing the entire vector of $T$ outputs satisfies differential privacy.
The main challenge for privacy is that each individual record contributes to outputs at multiple time steps.

Dwork et al.~\cite{DworkNPR10} and Chan et al.~\cite{ChanSS10} considered the problem of computing summation in the continual release model when each record consists of one bit. They designed a continual release mechanism, called the {\em binary tree mechanism}, that achieves (additive) error $O(\log^2 T)$  
for this problem.
Dwork et al.~\cite{DworkNPR10} also showed that an error of $\Omega(\log T)$ is necessary to privately release all running sums. (Further related work is discussed in \Cref{sec:related-work}.)

\subsection{Our Contributions}

We ask what price 
\ifnum\pods=0
differentially private 
\else
DP
\fi
algorithms must pay in accuracy to solve a problem in the continual release model instead of the batch model.
%
 The largest previously known gap in accuracy between the two models is logarithmic in $T$, exhibited by the result of~\cite{DworkNPR10} on summation.
We show that for two fundamental problems, which are related to summation and widely studied in the batch model, the gap is exponentially larger.

In the first problem, called $\maxsum{}$,
each input consists of $d$ binary attributes and the goal is to approximate the maximum of the attribute sums. We define the error of a mechanism as the maximum error over all the time steps. For $\maxsum{}$, the error at each time step is the absolute value of the difference between the true answer and the output of the mechanism at that time step. 
The second problem, $\selection{}$, is the ``argmax” version of $\maxsum{}$: the goal is to find the index of the largest attribute sum. The error at a particular time step for this problem is the absolute difference between the maximum sum and the attribute sum at the index returned by the mechanism at that time step. 
%
Both problems are abstractions of practically relevant tasks.
For instance, if the data collected by a public health agency (e.g., the US CDC) consists of records indicating which of $d$ medical conditions each person suffers from, then $\maxsum{}$ corresponds to the number of cases of the most common condition that occurred so far, and $\selection{}$ corresponds to the name of this condition.
Algorithms for these tasks are key ingredients in differentially private solutions to more complex problems such as synthetic data generation \cite{HardtLM12} and high-dimensional optimization~\cite{TalwarTZ15}.
We prove tight bounds on the error for these two problems in the continual release model in terms of the parameters $T$, called the {\em time horizon}, and $d$, called the dimension, discussed above,
as well as the privacy parameter~$\eps$. 

To provide a comparison to the continual release model, we assume here that algorithms in the batch model get input datasets of size $T$. Intuitively, a batch algorithm $\alg$ is $(\eps,\delta)$-differentially private if, for all datasets $\dstreamvec$ and $\dstreamvec'$ that differ in one record, all events under the distributions $\alg(\dstreamvec)$ and $\alg(\dstreamvec')$ have similar probabilities. In the case of $\delta=0$ (also referred to as {\em pure differential privacy)}, these probabilities differ by at most a factor of $e^\eps$.  In the case of $\delta>0$ (referred to as {\em approximate differential privacy}), if these probabilities are $p$ and $p'$, they must satisfy $p\leq e^\eps \cdot p'+\delta.$ (See Definition~\ref{def:differentially private}). 
To give a meaningful privacy guarantee, the parameter $\delta$ has to be small: in our case, $\delta=o(\eps/T)$. 
For continual release mechanisms, we study \emph{event-level} privacy, where each user's data appears in a single record, as opposed to {\em user-level} privacy, where a user's data could be distributed over multiple records. (See~\cite{DworkNPR10} for the discussion of these two variants.)


We demonstrate a strong separation between the continual release and the batch models. For approximate differential privacy, we show that when $d$ is sufficiently large, $\maxsum{d}$ and $\selection{d}$ require $\tilde \Omega(T^{1/3})$  and $\tilde \Omega\bparen{(\frac{T}{\log d})^{1/3}}$ error blowup, respectively, in the continual release model compared to the batch setting. For pure differential privacy, the blowup (when $d$ is large) is $\tilde \Omega(T^{1/2})$ for $\maxsum{}$ and $\tilde \Omega(\frac{T^{1/2}}{\log^{1/2} d})$ for $\selection{}$. 

\ifnum\pods=0
\begin{table}[t] 
\else
\begin{table*}[t]
\vspace{4pt}
\fi
    \begin{center} 
\begin{tabular}{|c||l|l|l| } 
      \hline
 & Approximate DP ($\delta>0)$ 
 & Pure DP ($\delta=0)$ 
 & Reference
 \\
 \hline
 \hline
 \begin{tabular}{c} \vspace{5pt}\\$\maxsum{}$\\ \end{tabular}  
 & 
 $\tilde \Omega\Big(\min\Big\{\sqrt[3]{\frac T {\eps^2}},\frac{\sqrt{d}}\eps,T\Big\}\Big)$
&
$\tilde \Omega\Big(\min\Big\{\sqrt{\frac T \eps},\frac{d}\eps,T\Big\}\Big)$
& 
\ifnum\pods = 0
\small
\else
\fi
{Thm.~\ref{thm:main_maxsum}}

 \\ 
 & 
 $\tilde O\Big(\min\Big\{\sqrt[3]{\frac T {\eps^2}},\frac{\sqrt{d}\  { \color{blue}\text{polylog}(T)}}\eps,T\Big\}\Big)$
 &
$\tilde O\Big(\min\Big\{\sqrt{\frac T \eps},\frac{d\ { \color{blue}\text{polylog}(T)}}\eps,T\Big\}\Big)$ 
& 
\ifnum\pods = 0
\small
\else
\fi
{Cor.~\ref{cor:upper-bounds},~\ref{cor:pure_dp}}
\\
\hline
\begin{tabular}{c} \vspace{5pt}\\$\selection{}$\\ \end{tabular} & 
$\tilde \Omega\Big(\min\Big\{\sqrt[3]{\frac {T\log^2 d} {\eps^2}},\frac{\sqrt{d}}\eps,T\Big\}\Big)$ &
$\tilde \Omega\Big(\min\Big\{\sqrt{\frac {T\log d} {\eps}},\frac{d}\eps,T\Big\}\Big)$ 
&
\ifnum\pods = 0
\small
\else
\fi
{Thm.~\ref{thm:selection-LB}}

\\ 
&$\tilde O\Big(\min\Big\{\sqrt[3]{\frac {T\log^2 d} {\eps^2}},\frac{\sqrt{d}\ {\color{blue}\text{polylog}(T)}}\eps,T\Big\}\Big)$
&$\tilde O\Big(\min\Big\{\sqrt{\frac {T\log d} {\eps}},\frac{d\ { \color{blue}\text{polylog}(T)}}\eps,T\Big\}\Big)$
&
\ifnum\pods = 0
\small
\else
\fi
{Cor.~\ref{cor:upper-bounds},~\ref{cor:pure_dp}}
\\
 \hline
\end{tabular}
\end{center}
\caption{Our results on the error of $(\eps,\delta)$-DP mechanisms in the continual release model. The corresponding upper and lower bounds differ only in the $\text{polylog}(T)$ terms, highlighted in blue. For approximate differential privacy, the lower bounds apply when $\delta=o(\eps/T)$, and the upper bounds apply when $\delta>\text{poly}\bparen{\tfrac 1 T}$.}    \label{tab:results}
\ifnum\pods=1
\end{table*}
\else
\end{table}
\fi

Our results are summarized in Table~\ref{tab:results}.
To put our bounds in context, observe that for both problems we consider, there is a trivial algorithm that ignores its data, always outputs the same value and has error at most $T$ (since each attribute sum is an integer between 0 and $T$).
Our bounds on the error should be contrasted with the error achievable by $(\eps,0)$-differentially private algorithms in the batch model: $O(\frac 1\eps)$ for $\maxsum{}$, and $O\paren{\frac{\log d}\eps}$ for $\selection{}$. The former is obtained by an instantiation of the Laplace mechanism from~\cite{dwork2006calibrating} and the latter---by an instantiation of the exponential mechanism of McSherry and Talwar~\cite{McTalwar}.

 We obtain our lower bounds by reductions from problems in the batch model. The key is to consider tasks for which multiple instances of the same base problem on one dataset need to be solved. For $\maxsum{}$, the corresponding task in the batch model is to output all marginals. Each marginal can be thought of as an instance of computing the (appropriately rescaled) sum of values in the corresponding coordinate. For $\selection{}$, the task in the batch model is based on solving independent instances of finding the largest marginal, each on its own subset of coordinates. We use the lower bounds for batch algorithms for these problems by Bun et al.~\cite{bunUV18}, Hardt and Talwar~\cite{HardtT10}, and Steinke and Ullman~\cite{US17}.

Each of our lower bounds is the minimum of three terms, corresponding to different parameter regimes.
Our lower bounds are matched (up to polylogarithmic factors in $T$ and $1/\delta$), in each regime, by two simple mechanisms and one trivial mechanism. The trivial mechanism always outputs an arbitrary value in the right range. The first simple mechanism is based on recomputing the value of the desired statistic (e.g., $\maxsum{}$) at regular intervals and providing the same answer until it is recomputed again. The second simple mechanism uses the binary tree mechanism to track all $d$ coordinates separately and takes the maximum (or, in the case of $\selection{}$, argmax) of the noisy values. The guarantees of these mechanisms for $\maxsum{}$, $\selection{}$, and general functions of sensitivity 1 are stated in Section~\ref{sec:upper-bounds}.
%
%
Together, our mechanisms and our lower bounds characterize the error for $\maxsum{}$ and $\selection{}$ up to polylogarithmic factors in $T$ and $1/\delta$ in all regimes.

Our lower bounds apply to the original continual release model of Dwork et al.~\cite{DworkNPR10} and Chan et al.~\cite{ChanSS10}. In this model, which we refer to as the {\em nonadaptive} setting, privacy is defined  for streams fixed in advance. However, our matching upper bounds hold even 
when privacy and accuracy are  required 
for adaptively selected streams. In the adaptive version of the model, each record in the stream is chosen by an adversary after it sees all the answers of the mechanism from the prior time steps. This model gives more power to the adversary and therefore places more stringent requirements on privacy and accuracy.
This model may be of independent interest.

\subsection{Further Related Work}\label{sec:related-work} 
\paragraph{Event-level privacy} Bolot et al.~\cite{BolotFMNT13} and Perrier et al.~\cite{PerrierAK19} extended the tree mechanism of Dwork et al.~\cite{DworkNPR10} to work for weighted sums with exponentially decaying coefficients and for sums of bounded real values, respectively.
Song et al.~\cite{SongLMVC18} generalized the model to graph data and obtained a continual release mechanism 
for graph statistics, such as the degree distribution and subgraph counts, on bounded degree graphs. Fichtenberger et al.~\cite{FichtenHO21} studied a variety of other graph problems in the continual release setting, including minimum cut and densest subgraph. Differentially private online learning is investigated in a sequence of works~\cite{jainkt12, SmithT13, AgarwalS17} that use the summation primitive developed by Dwork et al.\ to obtain sublinear regret guarantees for many hypothesis classes. The adaptive continual release model arises implicitly in those works, but to our knowledge it was not formulated explicitly. 
Cardoso and Rogers~\cite{CardosoR21}  study, among other problems, SumSelect (called \textit{top-1 selection with unrestricted $\ell_0$ sensitivity} in their work) in the continual release model. Their focus is on empirical performance on streams that arise in practice, in which the index of the largest sum changes seldom. The recomputation-based algorithm we present 
for SumSelect can be seen as a special case of their KnownBase algorithm. They evaluate the accuracy of the algorithm empirically whereas our work provides theoretical bounds on the error.
\ifnum\pods=0
One of the contributions of~\cite{CardosoR21} is making the algorithms work in a more restrictive computational model, in which the algorithm only stores the current values of the sums at any given time step and the seed of a pseudorandom function. The algorithms we present here can also be implemented in their model using the techniques in their paper.
\fi

\paragraph{User-level differential privacy} User-level privacy in the \emph{continual release model} was first studied by Dwork et al.~\cite{DworkNPR10} and Chan et al.~\cite{ChanSS10}. \emph{User-level} privacy is more stringent than \emph{event-level} privacy, so the lower bounds in our paper apply directly to that model. Even though, in general, event-level privacy does not imply user-level privacy, the recomputation technique used in some of our algorithms gives \emph{user-level} privacy whenever the mechanism employed for the recomputations is user-level private.

\paragraph{Pan-Privacy}
Pan-privacy, defined by Dwork et al.~\cite{DworkNPRY10}, is a 
 model that protects against intrusions into the memory of the algorithm as it processes a stream. In pan-privacy, as in continual release, the input is presented as a stream. However, the requirement of pan-privacy is orthogonal to that of continual release; see \cite{DworkNPRY10} for details.

\ifnum \pods=1
\subsection{Organization}\label{sec:organization}
In Section~\ref{sec:prelims}, we define differential privacy and the continual release model with nonadaptively chosen inputs, and state the problems we consider. The remaining preliminaries can be found in Appendix~\ref{app:preliminaries}. Sections~\ref{sec:maxsum-lowerbound}--\ref{sec:upper-bounds} present our technical results and the definition of the adaptive continual release model. All omitted proofs appear in the appendix. 
\fi


\section{Definitions}\label{sec:prelims}

\subsection{Preliminaries on Differential Privacy}
We first introduce the notion of $(\epsilon, \delta)$-indistinguishability.
\begin{definition}[$(\epsilon, \delta)$-Indistinguishability]
Random variables $R_1$ and $R_2$  over the same outcome space $\mathcal{Y}$ are  {\em $(\epsilon, \delta)$-indistinguishable} (denoted $R_1 \approx_{\eps, \delta} R_2$) if for all subsets $S \subseteq \mathcal{Y}$,  
\ifnum\pods=0
the following hold:
\fi
\begin{align*}
\Pr[R_1 \in S] \leq e^{\epsilon} \Pr[R_2 \in S] + \delta;\\
\Pr[R_2 \in S] \leq e^{\epsilon} \Pr[R_1 \in S] + \delta.
\end{align*}
\end{definition}

A dataset $\vec{\dstream} = (\dstream_1, \dots, \dstream_n) \in \X^n$ is a vector of elements, called {\em records}, from a universe $\X$. Two datasets are {\em neighbors} if they differ in one record (i.e., one coordinate). Informally, differential privacy requires that an algorithm's output distributions are similar on all pairs of neighboring datasets. In the batch model, the algorithm receives datasets as one batch as opposed to in an online fashion.

\begin{definition}[Differential Privacy in Batch Model~\cite{dwork2006calibrating,DworkKMMN06}]\label{def:differentially private} A randomized algorithm $\alg: \X^n \rightarrow \mathcal{Y}$ is {\em $(\eps, \delta)$-differentially private (DP)} if for every pair of neighboring datasets $\vec{\dstream}, \vec{\dstream}'\in \X^n$,
 \begin{equation*}
    \alg(\vec{\dstream}) \approx_{\eps, \delta}   \alg(\vec{\dstream}').
 \end{equation*}
The case $\delta=0$ is referred to as {\em pure} differential privacy, whereas the case $\delta>0$ is called  {\em approximate} differential privacy.
 \end{definition} 

\ifnum\pods=0
 Differential privacy protects groups of individuals.

\begin{lemma}[Group Privacy~\cite{dwork2006calibrating}]\label{lem:group_privacy} Every $(\eps, \delta)$-DP algorithm \alg is $\left(\ell \eps, \delta' \right)$-DP for groups of size $\ell$, where $\delta' = \delta\frac{e^{\ell \eps} -1}{e^\eps-1}$; that is, for all datasets $\vec{\dstream}, \vec{\dstream}'$ such that $\|\vec{\dstream} - \vec{\dstream}' \|_0 \leq \ell$,
\begin{equation*}
    \alg(\vec{\dstream}) \approx_{\ell \eps, \delta'}   \alg(\vec{\dstream}').
\end{equation*}
\end{lemma}

Differential privacy is closed under post-processing.
\begin{lemma}[Post-Processing~\cite{dwork2006calibrating,BunS16}]\label{prelim:postprocess} If $\alg$  is an $(\eps, \delta)$-DP algorithm with output space $\mathcal{Y}$ and $\mathcal{B}$  is a randomized map from $\mathcal{Y}$ to $\mathcal{Z}$, then the algorithm $\mathcal{B} \circ \alg$ is $(\eps, \delta)$-DP.
\end{lemma}

\begin{definition}[Sensitivity] Let $f: \X^n \rightarrow \R^m$ be a function. Its $\ell_1$-sensitivity is
\begin{equation*}
    \max_{\text{neighbors } \vec{\dstream}, \vec{\dstream}' \in \X^n} \|f(\vec{\dstream}) - f(\vec{\dstream}')\|_1.
\end{equation*}
To define $\ell_2$-sensitivity, we replace the $\ell_1$ norm with the $\ell_2$ norm.
\end{definition}

Our algorithms use the standard Laplace mechanism to ensure differential privacy. 

\begin{definition}[Laplace Distribution] The Laplace distribution with parameter $b$ and mean $0$, denoted  $\Lap(b)$,
has probability density
\begin{equation*}
    h(r) = \frac{1}{2b}e^{-\frac{|r|}{b}} \text{ for all $r \in \mathbb{R}$}.
\end{equation*}
\end{definition}

\begin{lemma}[Laplace Mechanism]\label{prelim:laplace_dp} Let $f : \X^n \rightarrow \mathbb{R}^m$ be a function with $\ell_1$-sensitivity at most $\Delta_1$. Then the Laplace mechanism is algorithm
\begin{equation*}
    \alg_f(\vec{\dstream}) = f(\vec{\dstream}) + (Z_1, \ldots, Z_m),
\end{equation*}
where $Z_i \sim \Lap\left(\frac{\Delta_1}{\eps}\right)$. Algorithm $\alg_f$ is $(\eps, 0)$-DP.
\end{lemma}

\begin{lemma}[Exponential Mechanism \cite{McTalwar}]\label{lem:expmech}
Let $L$ be a set of outputs and $g: L \times \X^n \to \mathbb{R}$ be a function that measures the quality of each output on a dataset. Assume that for every $m \in L$, the function $g(m,.)$ has $\ell_1$-sensitivity at most $\Delta$. Then, for all $\eps,n > 0$ and for all datasets $\dset \in \X^n$, there exists an $(\eps, 0)$-DP mechanism that outputs an element $m\in L$ such that, for all $a>0$, we have
\begin{equation*}
    \Pr[\max_{i \in [L]} g(i,\dset) -  g(m,\dset) \geq 2\Delta \frac{(\ln |L| + a)}{\eps}] \leq e^{-a}. 
\end{equation*}
\end{lemma}

\begin{definition}[Gaussian Distribution] The Gaussian distribution with parameter $\sigma$ and mean 0, denoted $\Gauss(0,\sigma^2)$, has probability density
\begin{equation*}
    h(r) = \frac{1}{\sigma \sqrt{2\pi}} e^{-\frac{r^2}{2\sigma^2}}  \text{ for all $r \in \mathbb{R}$}.
\end{equation*}
\end{definition}

\subsection{Preliminaries on $\rho$-zCDP}\label{sec:zCDP}

This section contains preliminaries about  ``zero-concentrated differential privacy" (zCDP). The difference between zero-concentrated differential privacy and $(\eps,\delta)$-differential privacy is that zCDP requires output distributions on all pairs of neighboring datasets to be $\rho$-close (\Cref{def:rho-indistinguishable}) instead of $(\eps,\delta)$-indistinguishable.
In \Cref{sec:upper-bounds} we analyse the privacy of our upper bounds in terms of zCDP and then use the fact that zCDP implies $(\eps,\delta)$-differential privacy (\Cref{prop:CDPtoDP}) to compare our upper and lower bounds.

\begin{definition}[R\'enyi Divergence \cite{Renyi61}]
Let $Q$ and $Q'$ be distributions on $\mathcal{Y}$. For $\xi \in (1,\infty)$, the R\'enyi divergence of order $\xi$ between $Q$ and $Q'$(also called the $\xi$-R\'enyi Divergence) is defined as
\begin{align}
    D_{\xi}(Q \| Q') = \frac{1}{\xi-1} \log\left( \E_{r \sim Q'} \left[ \left(\frac{Q(r)}{Q'(r)}\right)^{\xi-1} \right]  \right).
\end{align}
Here $Q(\cdot)$ and $Q'(\cdot)$ denote either probability masses (in the discrete case) or probability densities (when they exist). More generally, one can replace  $\frac{Q(.)}{Q'(.)}$ with the the Radon-Nikodym derivative of $Q$ with respect to $Q'$.
\end{definition}

\begin{definition}[$\rho$-Closeness]\label{def:rho-indistinguishable}
Random variables $R_1$ and $R_2$ over the same outcome space $\mathcal{Y}$ are  {\em $\rho$-close} (denoted $R_1 \simeq_{\rho} R_2$) if for all $\xi \in (1,\infty)$, 
\begin{align*}
D_{\xi}(R_1\|R_2) \leq \xi\rho \text{ and }  D_{\xi}(R_2\|R_1) \leq \xi\rho,
\end{align*}
where $D_{\xi}(R_1\|R_2)$ is the $\xi$-R\'enyi divergence  between the distributions of $R_1$ and $R_2$.
\end{definition}

\begin{definition}[zCDP in Batch Model~\cite{BunS16}]
A randomized batch algorithm $\alg : \X^n \to \mathcal{Y}$ is $\rho$-zero-concentrated differentially private ($\rho$-zCDP), if, for all neighboring datasets $\vec{\dset},\vec{\dset}' \in \mathcal{X}^n$,
$$\alg(\vec{\dset}) \simeq_{\rho} \alg(\vec{\dset'}).$$
\end{definition}

One major benefit of using zCDP is that this definition of privacy admits a clean composition result. We use it when analysing the privacy of the algorithms in \Cref{sec:upper-bounds}.

\begin{lemma}[Composition \cite{BunS16}] \label{lem:cdp_composition}
Let $\alg : \mathcal{X}^n \to \mathcal{Y}$ and $\alg' : \mathcal{X}^n \times \mathcal{Y} \to \mathcal{Z}$ be batch algorithms. Suppose $\alg$ is $\rho$-zCDP and $\alg'$ is $\rho'$-zCDP. Define batch algorithm $\alg'' : \mathcal{X}^n \to \mathcal{Y} \times \mathcal{Z}$ by $\alg''(\vec{\dset}) = \alg'(\yvec,\alg(\yvec))$. Then $\alg''$ is $(\rho+\rho')$-zCDP.
\end{lemma}

The \emph{Gaussian mechanism} is used in \Cref{sec:upper-bounds}. It estimates a real-valued function on a database by adding Gaussian noise to the value of the function. 

\begin{lemma}[Gaussian Mechanism \cite{BunS16}] \label{prop:gaussian-mech}
Let $f : \X^n \to \mathbb{R}$ be a function with $\ell_2$-sensitivity at most $\Delta_2$. Let $\alg$ be the batch algorithm 
that, on input $\yvec$, releases a sample from $\mathcal{N}(f(\yvec), \sigma^2)$. Then $\alg$ is $(\Delta_2^2/2\sigma^2)$-zCDP.
\end{lemma}

The final lemma in this section relates zero-concentrated differential privacy to $(\eps,\delta)$-differential privacy.

\begin{lemma}[Conversion from zCDP to DP \cite{BunS16}]\label{prop:CDPtoDP}
For all $\rho,\delta > 0$, if batch algorithm $\alg$ is $\rho$-zCDP, then $\alg$ is $(\rho+2\sqrt{\rho \log(1/\delta)},\delta)$-DP.
\end{lemma}

\else
\fi

\ifnum\pods = 0
\subsection{The Continual Release Model  with Nonadaptively Chosen Inputs}\label{sec:definitions-continual-release}
\else
\subsection{Nonadaptive Continual Release Model}\label{sec:definitions-continual-release}
\fi

A {\em mechanism} in the continual release model \cite{DworkNPR10,ChanSS10} is an algorithm that receives its input  $\vec{\dstream}=(\dstream_1,\dots,\dstream_T) \in \X^T$ as a stream. At each time step  $t\in[T]$, it gets a record $\dstream_t$ and outputs an answer $a_t$. The output stream $(a_1,\dots,a_T)$ is denoted by $\avec$. We use $\xt = (\dstream_1,\dots,\dstream_t)$ for $t\in[T]$ to denote the first $t$ records in a stream $\dstreamvec$ (similarly, $\avec_{[t]}=(a_1,\dots,a_t)$.) The total number of records in the stream, denoted by $T$, is called the  {\em time horizon}. For simplicity, we assume $T$ is known to the mechanism.

We consider two variants of the continual release model. In the {\em nonadaptive} model of \cite{DworkNPR10,ChanSS10}, the input stream $\dstreamvec$ is fixed before the mechanism runs. The {\em adaptive} model, defined in Section~\ref{sec:defs-adaptive}, allows an adversary to choose each input record $\dstream_t$ for $t\in\{2,\dots,T\}$ based on the previous outputs $a_1,\dots,a_{t-1}$ of the mechanism. The adaptive model gives the adversary more power. Therefore, the nonadaptive model provides weaker guarantees in terms of both privacy and accuracy. 
All our lower bounds are for the nonadaptive model and, consequently, imply the same lower bounds for the adaptive model. In contrast, all our algorithmic results are for the adaptive model (and, consequently, they also hold in the nonadaptive model).

We refer to standard algorithms that get their input in one batch and produce one output as {\em batch} algorithms. For clarity, we refer to continual release algorithms as mechanisms.


\paragraph{Accuracy}  
We start by defining how well a given output approximates the value of a function. We use a notion of error that depends on the function. Given a function $f: \X^*\to \mathcal{Y}$, a dataset $\dstreamvec\in\X^*$, and an answer $a\in \mathcal{Y}$, let $\err_f(\dstreamvec, a)$ be a nonnegative number that quantifies how far off $a$ is from $f(\dstreamvec)$. Specifically, when $\mathcal{Y}=\R^k$,
\begin{align}\label{eq:error-real-valued}
    \err_f(\dstreamvec,a)=\|f(\dstreamvec)-a \|_\infty.
\end{align}
Later (in (\ref{eq:err_sumselect})), we define a different notion of error for the optimization problem \selection{}. 
\ifnum\pods=1
The
\else
Intuitively, the 
\fi
error for an optimization problem corresponds to the deficit in the objective function. 

\begin{definition}[Accuracy of a Mechanism]\label{def:na-accuracy}
In the \CR{}, a mechanism $\mech$ is {\em $(\alpha,T)$-accurate} for $f$  if, for all fixed input streams $\dstreamvec =(\dstream_1,\dots,\dstream_T)$, the maximum error $\err_f(\xt, a_t)$ over the outputs $a_1,\dots,a_T$ of mechanism $\mech$ is bounded by $\alpha$ with high probability, that is,
$$\quad \Pr_{\text{coins of }\mech}\left[ \max_{t\in[T]} \err_f(\xt, a_t)  \leq \alpha\right] \geq \frac{2}{3}.$$ 
\end{definition}



\paragraph{Privacy} Finally, we define privacy in the \CR{}. 
\begin{definition}[Privacy of a Mechanism]\label{def:na-privacy}
Given a mechanism $\mech$, define $\alg_\mech$ to be the batch model algorithm that receives an input dataset $\dstreamvec$, runs $\mech$ on stream $\dstreamvec$, and returns the output stream $\avec$ of $\mech$. The mechanism $\mech$ is {\em $(\epsilon, \delta)$-differentially private (DP) in the \CR{}} if $\alg_{\mech}$ is $(\epsilon, \delta)$-DP in the batch model.
\end{definition}

Definition~\ref{def:na-privacy} refers to \emph{event-level} privacy, where each user's data appears in a single record, as opposed to {\em user-level} privacy, where a user's data could be distributed over multiple records.

\subsection{Problem Definitions}

We consider two  functions on datasets, where each record consists of $d$ binary attributes. The first function, $\maxsum d$, returns the maximum attribute sum for the input records. The second function, $\selection d$, returns the index of such a maximum sum.
\begin{definition}\label{def:maxsum}
Let $d\in\N$ and $\X = \zo^d$.
For a dataset $\dstreamvec\in\X^*$ and $j\in[d]$, the {\em $j^{th}$ attribute} of record $\dstream_i$ is its $j^{th}$ coordinate, denoted $\dstream_i[j]$.
Let $t\in\N$ and $\xt\in \X^t$.
The function  $\maxsum{d} : \X^* \to \N$ is
$$\maxsum{d}(\xt) \defeq \max_{j\in[d]}\Big(\sum_{i\in[t]} \dstream_i[j]\Big).$$
The function  $\selection{d} : \X^* \to [d]$ is 
$$\selection{d}(\xt) \defeq \argmax_{j\in[d]}\Big(\sum_{i\in[t]} \dstream_i[j]\Big).$$
If multiple indices $j$ attain the maximum sum, the function value is defined to be the smallest such index.
\end{definition}





We study the accuracy of differentially private algorithms for computing these two functions. 
Our accuracy goal, stated in Definition~\ref{def:na-accuracy}, uses the notion $\err_f$. We define the error $\err_{\maxsum{}}$ as in (\ref{eq:error-real-valued}). For \selection{}, it is defined by:
\begin{equation}\label{eq:err_sumselect}
\err_{\selection{}} (\xt, a_t) = 
\maxsum{d}(\xt) -\sum_{i\in[t]} x_i[a_t].
\end{equation}

\section{Lower Bounds for $\maxsum{}$}\label{sec:maxsum-lowerbound}

In this section, we prove \Cref{thm:main_maxsum} that provides strong lower bounds on the accuracy parameter $\alpha$ for any accurate mechanism
for $\maxsum{d}$ in the \CR{}. Our lower bounds match the upper bounds from \Cref{sec:upper-bounds} for $\maxsum{d}$ in the \aCR{} up to logarithmic factors in the time horizon $T$ and the number of coordinates $d$.

\begin{theorem}\label{thm:main_maxsum}
For all $\eps \in (0,1], \delta\in[0,1),\alpha \geq 0, d\in\N$, sufficiently large $T\in\N$, and mechanisms $\mech$ in the \CR{} that are $(\varepsilon,\delta)$-differentially private and $(\alpha, T)$-accurate for $\maxsum{d}$, the following statements hold.
\begin{enumerate}
\item If $\delta > 0$ and $\delta = o(\eps/T)$, then 
    \ifnum\pods=1
    \\
    \else
    \fi
    $\alpha = \Omega\Big( \min\left\{ \frac{T^{1/3}}{\eps^{2/3}\log^{2/3}( \eps T)},  \frac{\sqrt{d}}{\eps \log d}, T  \right\} \Big)$.
\item If $\delta = 0$, then $\alpha = \Omega \Big(\min\left\{   \sqrt{\frac{T}{\eps}}, \frac{d}{\eps}, T \right\} \Big)$.
\end{enumerate}
\end{theorem}

$\maxsum{d}$ can be released in the batch model with $\alpha = O(1/\eps)$ via the Laplace mechanism \cite{dwork2006calibrating}. Hence, \Cref{thm:main_maxsum} shows a strong separation between the batch model of differential privacy and the continual release model.


\subsection{1-way Marginal Queries in Batch Model}
To prove our lower bounds for $\maxsum{}$, we reduce from the problem of approximating 1-way marginals in the batch model.
The function $\marginals{d}: \X^* \to [0,1]^d$ maps a dataset $\vec{\dset}$ of any size $n$ to a vector $(q_1(\yvec),\dots,q_d(\yvec)),$ where $q_j$, called the $j^{th}$ marginal, is defined as %
$q_j(\vec{\dset}) = \frac{1}{n}\sum_{i=1}^n \vec{\dset}[j].$
The error $\err_{\marginals{}}$ is defined as in (\ref{eq:error-real-valued}). Next, we define accuracy for batch algorithms.

\begin{definition}[Accuracy of Batch Algorithms]\label{def:batchmodel-accuracy}
Let $\gamma \in [0,1]$, $n,d\in\N$, and $\X = \zo^d$. Let $f: \X^n \to \R^d$ be a function on datasets. Batch algorithm 
$\alg$ 
is {\em $(\gamma,n)$-accurate} for $f$ if for all datasets $\vec{\dset}\in\X^n$,
$$\Pr_{{\text{coins of }}\alg}\left[ \err_f(\vec{\dset}, \alg(\vec{\dset})) \leq \gamma\right] \geq \frac{2}{3} . $$
\end{definition}

We use the lower bounds from \cite{bunUV18,HardtT10} for the problem of estimating $\marginals{d}$ in the batch model. They are stated in Items~1 and 2 of \Cref{lem:oneway-marginals} for
approximate differential privacy and pure differential privacy, respectively. 
Item~2 in \Cref{lem:oneway-marginals} is a slight modification of the lower bound from \cite{HardtT10} and follows from a simple packing argument.
\begin{lemma}\label{lem:oneway-marginals}
For all $\eps \in (0,1]$, $\delta \in [0,1]$, $\gamma \in (0,1)$, $d,n \in \N$,  and algorithms $\alg$ that are $(\eps, \delta)$-differentially private and $(\gamma,n)$-accurate for $\marginals{d}$, the following statements hold.\\ 
\hspace*{3mm} {\bf 1} {\em (\cite{bunUV18}).} If $\delta > 0$ and $\delta = o(1/n)$, then $n= \Omega\left( \frac{\sqrt{d}}{\gamma\eps  \log d} \right)$.\\
\hspace*{3mm}  {\bf 2} {\em (\cite{HardtT10}).} If $\delta = 0$, then $n = \Omega\left( \frac{d}{\gamma \epsilon} \right)$.
\end{lemma}


\subsection{Proof of \Cref{thm:main_maxsum}}

Let $\mech$ be an $(\eps,\delta)$-DP and $(\alpha, T)$-accurate mechanism for $\maxsum{d}$ in the \CR. We use $\mech$ to construct an $(\eps, \delta)$-DP batch algorithm $\alg$ that is $(\frac{\alpha}{n},n)$-accurate  for $\marginals{d}$.
The main idea in the construction, presented in Algorithm~\ref{alg:maxsum_reduction}, is to force $\mech$ to output an estimate of the sum for one attribute at a time by making the sum in that attribute the largest. First, $\alg$ streams its own dataset $\yvec$ to $\mech$. Then it sends $n$ additional records with 1 in the first attribute and 0 everywhere else. After this, the first attribute sum is the largest, and the answer produced by $\mech$ at this point can be used to estimate the first marginal. Then $\alg$ equalizes the number of extraneous 1's for each attribute by sending $n$ additional records with 0 in the first attribute and 1 everywhere else. It repeats this for each attribute, collecting the answers from $\mech,$ and then outputs its estimates for the marginals.

For vectors $\vec{u}=(u_1,\dots,u_{\ell})$ and $\vec{v}=(v_1,\dots,v_m)$, let $\vec{u}\circ\vec{v}=(u_1,\dots,u_{\ell},v_1,\dots,v_m)$.
For a vector $\vec{v}$, let $\vec{v}^n$ denote the vector $\vec{v}\circ \vec{v} \circ\dots\circ \vec{v}$ representing $n$ concatenated copies of $\vec{v}.$

\begin{algorithm}[h]
    \caption{Algorithm $\alg$ for estimating all 1-way marginals}
    \label{alg:maxsum_reduction}
    \begin{algorithmic}[1]
        \Statex\textbf{Input:} $\vec{\dset} = (\dset_1,\dots,\dset_n) \in \X^n$, where $\X = \zo^d$, and black-box access to mechanism $\mech$.
        \Statex\textbf{Output:}  $\vec{b} = (b_1,\dots,b_d) \in \R^d$.
        \State Let $\vec{e}_j$ be a vector of length $d$ with $1$ in coordinate $j$ and $0$ everywhere else; let $\overline{\vec{e}_j} \gets (1)^d - \vec{e}_j$.
        \State \label{step:construct-stream}Construct a stream $\vec{\dstream} \gets \vec{\dset} \circ (\vec{e}_1)^n \circ (\overline{\vec{e}_1})^n \circ \dots \circ (\vec{e}_{d-1})^n \circ (\overline{\vec{e}_{d-1}})^n \circ (\vec{e}_d)^n$ with $2dn$ records.\label{line:allz_nonad}
        \For{$t\in [T]$}
        \State Send $\dstream_t$ to $\mech$ and get the corresponding output $a_t$.
        \EndFor
        \For{$j\in[d]$}
        \State $b_j \gets a_{2jn}/n - j$.  \label{algline:maxsum_setoutput}
        \EndFor
        \State Output $\vec{b} \gets (b_1,\dots,b_d)$. 
    \end{algorithmic}
\end{algorithm}
\begin{lemma}\label{lem:maxsum_reduction}
Let $\alg$ be Algorithm~\ref{alg:maxsum_reduction}. For all  $\eps > 0, \delta \geq 0, \alpha \in \R^{+}$ and $d,n,T \in \N$, where $T \geq 2dn$,  if  mechanism $\mech$  is $(\eps, \delta)$-DP and $(\alpha,T)$-accurate for $\maxsum{d}$ in the \CR{}, then batch algorithm~$\alg$ 
is $(\eps, \delta)$-DP and $(\frac{\alpha}{n},n)$-accurate for $\marginals{d}$.
\end{lemma}

\begin{proof}
We start by reasoning about privacy. Fix neighboring datasets $\vec{\dset}$ and $\vec{\dset'}$ that are inputs to algorithm~$\alg$. Let $\vec{\dstream}$ and $\vec{\dstream}'$ be the streams constructed in Step~\ref{step:construct-stream} of $\alg$ when it is run on $\vec{\dset}$ and $\vec{\dset}'$, respectively. By construction, $\vec{\dstream}$ and $\vec{\dstream'}$ are neighbors. Since $\mech$ is $(\eps, \delta)$-DP, and $\alg$ only post-processes the outputs received from $\mech$, Lemma~\ref{prelim:postprocess} implies that $\alg$ is $(\eps, \delta)$-DP.

Now we reason about accuracy. Let $\vec{\dstream} = (\dstream_1,\dots,\dstream_{2dn})$ be the input stream provided to $\mech$ when $\alg$ is run on dataset $\vec{\dset}.$
By construction of $\dstreamvec$, the marginals $q_j(\yvec)$ for all $j\in[d]$
 and $\maxsum{d}$ are related as follows:
\ifnum\pods=0
\begin{equation}\label{eq:maxsum_obs}
    q_j(\vec{\dset}) = \frac{1}{n}\sum_{i\in[n]} \dset_i[j] 
    = \frac{1}{n} \Big(\sum_{i\in[2jn]} \dstream_i[j] -jn\Big)
    =  \frac 1 n \cdot \maxsum{d}(\dstreamvec_{[2jn]}) - j.
\end{equation}
\else
\begin{align}
     q_j(\vec{\dset}) 
     &= \frac{1}{n}\sum_{i\in[n]} \dset_i[j] 
     = \frac{1}{n} \Big(\sum_{i\in[2jn]} \dstream_i[j] -jn\Big) \nonumber\\
    & =  \frac 1 n \cdot \maxsum{d}(\dstreamvec_{[2jn]}) - j.\label{eq:maxsum_obs}
\end{align}
\fi

The attribute with the largest sum in $\dstreamvec_{[2jn]}$ is $j$ because $(\vec{e}_1)^n \circ (\overline{\vec{e}_1})^n \circ \dots \circ (\vec{e}_{j-1})^n \circ (\overline{\vec{e}_{j-1}})^n \circ (\vec{e}_j)^n$ contributes $jn$ ones to this attribute and $(j-1)n$ ones to each attribute in $[n]/\{j\}$, whereas the maximum sum of any attribute in $\yvec$ is $n$.

Since the transformation from $\mech$ to $\alg$ is deterministic, the coins of $\alg$ are the same as the coins of $\mech$. By~(\ref{eq:maxsum_obs}) and the computation of the estimates for the $\marginals{d}$ in Step~\ref{algline:maxsum_setoutput} of \Cref{alg:maxsum_reduction}, 
%
%
\ifnum\pods=1
\begin{align*}
\Pr_{{\text{coins of }}\alg}&\left[ \err_{\marginals{}}(\vec{\dset}, \alg(\vec{\dset})) \leq \frac \alpha n\right]\\
= &\Pr_{\text{coins of }\alg}\left[\max_{j\in[d]}\abs{q_j(\vec{\dset})-b_j} \leq \frac{\alpha}{n} \right] \\
= &\Pr_{\text{coins of }\mech}\left[ \max_{t\in\{2n,\dots,2dn\}} \abs{ \maxsum{d}(\xt) - a_t} \leq \alpha \right] \\
\geq &\Pr_{\text{coins of }\mech}\left[ \max_{t\in[T]} \abs{ \maxsum{d}(\xt) - a_t} \leq \alpha \right] \\
= &\Pr_{\text{coins of }\mech}\left[ \max_{t\in[T]} \err_{\maxsum{}}(\xt, a_t)  \leq \alpha\right]
 \geq \frac{2}{3},
\end{align*}
\else
\begin{align*}
& \Pr_{{\text{coins of }}\alg}\left[ \err_{\marginals{}}(\vec{\dset}, \alg(\vec{\dset})) \leq \frac \alpha n\right]
&&= \Pr_{\text{coins of }\alg}\left[\max_{j\in[d]}\abs{q_j(\vec{\dset})-b_j} \leq \frac{\alpha}{n} \right] \\
= &\Pr_{\text{coins of }\mech}\left[ \max_{t\in\{2n,\dots,2dn\}} \abs{ \maxsum{d}(\xt) - a_t} \leq \alpha \right] 
&&\geq \Pr_{\text{coins of }\mech}\left[ \max_{t\in[T]} \abs{ \maxsum{d}(\xt) - a_t} \leq \alpha \right] \\
= &\Pr_{\text{coins of }\mech}\left[ \max_{t\in[T]} \err_{\maxsum{}}(\xt, a_t)  \leq \alpha\right]
 &&\geq \frac{2}{3},
\end{align*}
\fi
%
where we used that $\mech$ is $(\alpha,T)$-accurate for $\maxsum{d}$. Thus, $\alg$ is $(\frac \alpha n,n)$-accurate for $\marginals{d}$.
\end{proof}

\ifnum\pods=0
Now, we are ready to prove \Cref{thm:main_maxsum}. 
\fi
\begin{proof}[Proof of \Cref{thm:main_maxsum}]
Observe that the accuracy parameter $\alpha$ is nondecreasing as a function of $d$, since a mechanism $\mech$ for $\maxsum{d}$ can be used to approximate $\maxsum{d'}$ for all $d'<d$ with the same accuracy and privacy guarantees by padding each length-$d'$ input record with $d - d'$ zeroes.

Recall that both  lower bounds on $\alpha$ stated in \Cref{thm:main_maxsum} are the minimum of three terms. To prove them,
it suffices to show that, for all ranges of parameters, one of the terms is a lower bound on $\alpha$.

First, consider the case when $\eps \leq \frac{2}{T}$. We will show that in this case (for both pure and approximate differential privacy), $\alpha>T/9.$ Since $\alpha$ is a nondecreasing function of $d$, it is sufficient to show this for $d=1$. Suppose for the sake of contradiction that $\alpha\leq T/9.$
 Let $\vec{\dstream}=(0)^T$ and $\vec{\dstream'}=(0)^{3T/4} (1)^{T/4}$ be datastreams that differ on $T/4$ records. 
Let $a_T$ and $a'_T$ be the final outputs of $\mech$ on input streams $\vec{x}$ and $\vec{x'}$, respectively. 
By accuracy of $\mech$, we have $\Pr[a_T\leq T/9]\geq 2/3.$
Applying Lemma~\ref{lem:group_privacy} on group privacy with $\eps \leq {2}/{T}$ and $\ell=T/4$, we get 
$\Pr[a'_T > T/9] \leq \sqrt{e}\cdot\Pr[a_t> T/9] + \frac{2\delta}{\eps}
<2/3$
for sufficiently large $T$, since $\delta = o(\eps / T)$. 
But $\maxsum d (\dstreamvec')=T/4$, so
$\mech$ is  not $(T/9,T)$-accurate, a contradiction. Hence, $\alpha = \Omega(T)$.

Now assume $\eps>\frac 2 T,$ i.e., $\eps T>2.$
We start by proving Item~1 (when $\delta = o(\eps/T)$). Let $\alg$ be the algorithm for $\marginals d$ with black-box access to $\mech$, as defined in \Cref{alg:maxsum_reduction}.
%
%
%
%
If $T \geq 2dn$ and $\frac{\alpha}{n}<1$, then by \Cref{lem:maxsum_reduction}, algorithm $\alg$ is $(\epsilon,\delta)$-differentially private and $(\frac{\alpha}{n},n)$-accurate  for $\marginals{d}$. (We require $\frac{\alpha}{n}<1$ for the accuracy guarantee on $\alg$ to be meaningful.) We can then use \Cref{lem:oneway-marginals} to lower bound $\alpha$.

{\bf Case 1: $d \leq (\eps T\log (\eps T))^{2/3}$.}
If there exists a dataset size $n\in(\alpha, \frac{T}{2d}]$, then by Item~1 of \Cref{lem:oneway-marginals}, $n=\Omega \big( \frac{n\sqrt{d}}{\alpha\cdot\epsilon \log d} \big)$, and hence $\alpha = \Omega \big( \frac{\sqrt{d}}{\epsilon \log d}\big)$. If no such $n$ exists, then $\alpha + 1 \geq \frac{T}{2d}$, and hence $\alpha = \Omega(\frac T d) = \Omega\big( \frac{T^{1/3}}{\eps^{2/3}\log^{2/3}( \eps T)} \big)$.
Combining the expressions for the two parameter ranges, we get that $\alpha = \Omega\big( \min\big\{ \frac{T^{1/3}}{\eps^{2/3}\log^{2/3}( \eps T)},  \frac{\sqrt{d}}{\eps \log d}  \big\} \big).$

{\bf Case 2: $d > (\eps T\log(\eps T))^{2/3}$.} Set $d'=\lfloor (\eps T\log(\eps T))^{2/3}\rfloor.$ Observe that $d\geq 1$ because $\eps T>2$. By our previous padding argument, a mechanism for $\maxsum{d}$  can be used to approximate $\maxsum{d'}$ for $d' = (\eps T)^{2/3}$ with the same accuracy and privacy guarantees. Therefore, 
\ifnum\pods=1
\begin{align*}
\alpha & = \Omega\left( \min\left\{ \frac{T^{1/3}}{\eps^{2/3}\log^{2/3}( \eps T)},  \frac{\sqrt{d'}}{\eps \log d'}  \right\} \right) \\& =  \Omega \left( \frac{T^{1/3}}{\epsilon^{2/3} \log^{2/3}(\eps T)} \right).
\end{align*}
\else
$\alpha = \Omega\left( \min\left\{ \frac{T^{1/3}}{\eps^{2/3}\log^{2/3}( \eps T)},  \frac{\sqrt{d'}}{\eps \log d'}  \right\} \right) =  \Omega \left( \frac{T^{1/3}}{\epsilon^{2/3} \log^{2/3}(\eps T)} \right).$
\fi
\ifnum\pods = 1
\\
\else
\fi
This completes the proof of Item~1.

The proof of Item~2 (with $\delta = 0$) proceeds along the same lines, except that we consider the cases $d \leq \sqrt{\eps T}$ and $d > \sqrt{\eps T}$ and use Item~2 from \Cref{lem:oneway-marginals} instead of Item~1.
If a dataset size $n\in(\alpha,\frac{T}{2d}]$ exists, by Item~2 of \Cref{lem:oneway-marginals}, we get $n=\Omega \left( \frac{nd}{\alpha\eps} \right)$, and hence $\alpha = \Omega \left( \frac{d}{\eps} \right)$.
If no such $n$ exists, then $\alpha + 1 \geq \frac{T}{2d}$, and hence $\alpha = \Omega(T/ d) = \Omega(\sqrt{{T}/{\eps}})$.
If $d>\sqrt{\eps T}$, a padding argument gives that $\alpha = \Omega(\sqrt{T / \eps})$.
\end{proof}

\section{Lower Bounds for $\selection{}$}\label{sec:selection-lowerbound}

In this section, we prove \Cref{thm:selection-LB} that provides strong lower bounds on the accuracy parameter $\alpha$ of any $(\alpha,T)$-accurate algorithm $\mech$ for $\selection{d}$ in the \CR{}.  Our lower bounds match the upper bounds from \Cref{sec:upper-bounds} for $\selection{d}$ in the \aCR{} up to logarithmic factors in the time horizon $T$ and the number of coordinates $d$.

\begin{theorem}\label{thm:selection-LB}
For all $\varepsilon \in (0,1], \delta\in[0,1), \alpha > 0,$ $d\in \N$ such that $d > 1$, sufficiently large $T\in\N$, and mechanisms
$\mech$ in the \CR{} that are $(\varepsilon,\delta)$-DP and $(\alpha, T)$-accurate for $\selection{d}$, the following statements hold.
\begin{enumerate}
    \item If $0<\delta = o(\frac{\eps}{T})$, then $\alpha =\tilde \Omega\Bparen{\min\left\{\frac{T^{1/3} \log^{2/3} d}{\eps^{2/3}}, \frac{\sqrt{d}}{\eps}, T\right\} }$.

\item   If $\delta = 0$, then 
 $\alpha = \Omega\Bparen{\min\left\{
\sqrt{\frac{T}{\eps}\log(2 + \frac{d}{\sqrt{\eps T}})}, 
 \frac {d}{\eps},
 T
 \right\}}
\allowbreak = \tilde \Omega \Bparen{\min\left\{\sqrt{\frac{T\log(d)}{\eps} }, \frac d \eps  ,T\right\}}$.
\end{enumerate}
\end{theorem}


\subsection{$\kindselection{d}$ Problem in the Batch Model}

To prove our lower bounds for $\selection{}$ in the \CR{}, we reduce from the problem called $\kindselection{}$ that solves $k$ disjoint instances of the problem of selecting the index of the largest marginal 
in the batch model.

To define the function $\kindselection{}$,
let $n,d,k \in \N$, and $\X = \zo^{kd}$. Let $\vec{\dset}[i:j]$ denote the dataset $\yvec \in \X^n$ with each record restricted to the coordinates between (and including) $i$ and~$j$. The function $\kindselection{d}: \X^{n} \to [d]^k$ corresponds to dividing the dataset into $k$ {\em blocks} $\vec{\dset}[1:d], \vec{\dset}[d+1:2d], \dots, \vec{\dset}[(k-1)d+1:kd]$, with $n$ records each, and applying $\selection{d}$ independently on each block. It maps a dataset $\vec{\dset}$ of size~$n$ to a vector $(h_1(\yvec),\dots,h_k(\yvec)),$ where $h_r$ is defined as the $\selection{d}$ function applied to block $r$:
$$h_r(\vec{\dset}) = \selection{d}\Big(\vec{\dset}\left[(r-1)d+1:rd\right]\Big).$$
The accuracy for $\kindselection{}$ is defined as in \Cref{def:batchmodel-accuracy}. To apply it, we define the error $\err_{\kindselection{}}$. Note that the error is scaled differently than for $\selection{}$ because the goal is to select the index of the largest marginal in each block, not of the largest sum.
For $\vec{b} = (b_1, \dots, b_k) \in [d]^k$, define $\err_{\kindselection{}}(\vec{\dset},\vec{b})$
\begin{align*}
      = \max_{r \in [k]} \left( \frac 1n\cdot\err_{\selection{}}(\yvec[r(d-1)+1:rd],b_r) \right).
\end{align*}
%
Next, we state lower bounds for $(\eps, \delta)$-differentially private approximation of $\kindselection{}$ in the batch model. 

\begin{lemma}\label{lem:kindsel-mainlb}
For all  $\eps \in (0,1]$, $\delta \in [0,1]$, $\gamma \in [0,\frac{1}{20}]$, $d, k, n \in \N$, and batch algorithms $\alg$ that are $(\eps, \delta)$-differentially private and $(\gamma,n)$-accurate for $\kindselection{d}$, the following statements hold.
\begin{enumerate}
    \item If $\delta > 0$ and $\delta = o(1/n)$, then $n = \Omega(\frac{\sqrt{k} \cdot \log d}{\eps\gamma \log (k+1)})$.
    \item If $\delta = 0$, then $n = \Omega\left( \frac{k\cdot \log d}{\eps\gamma} \right)$.
\end{enumerate}
\end{lemma}

\ifnum\pods=0

Item~1 in Lemma~\ref{lem:kindsel-mainlb} follows from \Cref{lem:lbkindsel} below.
\begin{theorem}[\cite{US17,Ullman21pers}]\label{lem:lbkindsel}
For all $\epsilon \in (0,1], \delta \in (0,1/n]$, $\gamma \in [0,\frac{1}{20}]$, $d,n,k \in \N$, if Algorithm $\alg$ is $(\eps,\delta)$-differentially private and $(\gamma,n)$-accurate for $\kindselection{d}$, then $n = \Omega(\frac{\sqrt{k}\log d}{\gamma \eps \log (k+1)})$.
\end{theorem}
\begin{proof}[Proof Sketch]
We are aware of two proofs of this result, both of which were communicated to us by Jonathan Ullman~\cite{Ullman21pers}. The first uses the top-$k$ selection lower bound of Steinke and Ullman~\cite{US17}. In that problem, there is a single collection of $d$ coordinates and the goal is to return the indices of $k<d$ coordinates whose sums are roughly largest. 

For the specific distribution over instances that arises in the lower bound of~\cite{US17}, if one divides the coordinates into $k$ equal groups, there is a constant probability that the collection of coordinates with the largest sum in each group is a good approximate solution for the top-$k$ selection problem. An algorithm for $\kindselection{d}$ can thus be used to solve the top-$k$ selection (out of $dk$ coordinates) problem for such instances with roughly the same error and privacy parameter. The lower bound of~\cite{US17} on $n$ then applies.

Another approach is to use the composition framework of Bun, Ullman and Vadhan \cite{bunUV18}. One can use a folklore result that selection among $d>2^m$ coordinates can be used to mount a reconstruction attack on an appropriate dataset of size $m$. Composed with the lower bound for 1-way marginals in \cite{bunUV18}, one obtains a lower bound for $\kindselection{d}$.
\end{proof}

To complete the proof of \Cref{lem:kindsel-mainlb}, we prove Item $2$ via a standard packing argument. 


\begin{proof}[Proof of Item~2 in \Cref{lem:kindsel-mainlb}.]
For $\vec{u}\in [d]^k,$ define $\dset^*_{\vec{u}} \in \zo^{dk}$ to be the record where each block $r\in[k]$ of $d$ coordinates has a $1$ in coordinate $u_r$ and all zeros everywhere else. Let $\vec{\dset_u}$ be the dataset that consists of $2\gamma n$ copies of~$\dset^*_{\vec{u}}$ and $(1-2\gamma)n$ copies of the all-zero record (assuming, for simplicity, that $2\gamma n$ is an integer). Since $\alg$ is $(\gamma,n)$-accurate, $\Pr_{\text{coins of \alg}}\left[ \err_{\kindselection{}}(\vec{\dset_u},\alg(\vec{\dset_u})) \leq \gamma \right] \geq \frac{2}{3}$ for all $\vec{u}\in[d]^k$. This means that for all $\vec{u}\in[d]^k$,
\[\Pr_{\text{coins of \alg}}\left[\alg(\vec{\dset_u}) = \vec{u}\right] \geq \frac{2}{3}.\]
For all $\vec{u},\vec{u}' \in [d]$, by group privacy,
$\alg(\vec{\dset_u}) \approx_{(\gamma\eps n,0)} \alg(\vec{\dset_{u'}})$, which implies that
\begin{align}\label{eq:groupprivlb}
\Pr\left[\alg(\vec{\dset_{u}}) = \vec{u}'\right] &\geq e^{-\gamma\eps n}\Pr\left[\alg(\vec{\dset_{u'}}) = \vec{u}'\right] \geq \frac{2}{3}e^{-\gamma\eps n}.
\end{align}
Since the probability of any event is at most 1,
\ifnum\pods=0
\[
1\geq \Pr_{\text{coins of \alg}}\left[ \alg(\vec{\dset_u}) \neq \vec{u} \right] = \sum_{\vec{u}'\neq \vec{u}}\Pr\left[\alg(\vec{\dset_u}) = \vec{u}'\right] \geq \frac{2}{3}e^{-\gamma\eps n}(d^k-1),
\]
\else
\begin{align*}
& 1\geq \Pr_{\text{coins of \alg}}\left[ \alg(\vec{\dset_u}) \neq \vec{u} \right] = \\
& \sum_{\vec{u}'\neq \vec{u}}\Pr\left[\alg(\vec{\dset_u}) = \vec{u}'\right] \geq \frac{2}{3}e^{-\gamma\eps n}(d^k-1),
\end{align*}
\fi

where the last inequality holds by (\ref{eq:groupprivlb}). We get that
$e^{\gamma\eps n} \geq \frac{d^k-1}{2}\cdot\frac{2}{3}$, and thus $n = \Omega\left(\frac{k\log d}{\gamma\eps}\right).$
\end{proof}
\else
\fi

\subsection{Proof of \Cref{thm:selection-LB}}
Let $\mech$ be an $(\eps,\delta)$-DP and $(\alpha, T)$-accurate mechanism for $\selection{}$ in the \CR. We use $\mech$ to construct an $(\eps, \delta)$-DP algorithm $\alg$ that is $(\frac{\alpha}{n},n)$-accurate  for $\kindselection{d}$ in the batch model.
We motivate our approach by first discussing an idea that doesn't quite work. Let $\mech$ be an accurate mechanism for \selection{d} in the \CR{} and  $\vec{\dset}$ be a dataset with $n$ records from $\zo^{dk}$. A naive approach to solving \kindselection{d}{} 
in the batch model is to run $k$ instantiations of $\mech$ for $n$ time steps each, one on each block of $d$ coordinates, to select the coordinate with the maximum sum in that block. However, running $k$ instantiations of $\mech$, as described, would result in a significant degradation of privacy, because every datapoint is used $k$ times, once for each instantiation of $\mech$. We instead reduce to $\selection{dk}$ and run a single instantiation of $\mech$ for about $nk$ time steps, where each datapoint in $\vec{\dset}$ is sent to $\mech$ only once. This approach doesn't suffer from privacy degradation. 

Algorithm $\alg$ proceeds in $k$ stages; the $r^{th}$ stage is dedicated to selecting the coordinate with the maximum sum in the $r^{th}$ block.
In the first stage, $\alg$ streams $\vec{\dset}$ to $\mech$. 
In order to select the coordinate with the maximum sum from the first block, $\alg$ then sends $2n$ records of the form $(1^d0^d\dots0^d)$ to $\mech$. Then the sums of the coordinates in the first block of $\vec{\dset}$ become much larger than the sums in the other blocks. This ensures that at the end of the first stage, $\mech$ selects the coordinate with the maximum sum in the first block. In the second stage, $\alg$ sends $2n$ records of the form $(0^d1^d\dots1^d)$ to $\mech$ in order to balance out the number of extraneous 1's for each coordinate.
In order to select the coordinate with the maximum sum from the second block, $\alg$  sends $2n$ records of the form $(0^d1^d 0^d \dots0^d)$ to $\mech$. At the end of the second stage, $\mech$ selects the coordinate with the maximum sum in the second block.  Algorithm $\alg$ proceeds similarly for every block.


\ifnum\pods=0 
The details of the algorithm appear in \Cref{alg:kindsel}.
\else
See \Cref{alg:kindsel} for details.
\fi
For ease of indexing, $\alg$ sends all-zero records in time steps $n+1$ to $2n$ in \Cref{line:allz_nonad} of \Cref{alg:kindsel}, to ensure that all stages have $4n$ time steps.

\begin{algorithm}[h]
    \caption{Batch algorithm $\alg$ for $\kindselection{}$}
    \label{alg:kindsel}
    
    \begin{algorithmic}[1]
        \Statex\textbf{Input:} $k$, $\vec{\dset} = (\dset_1,\dots,\dset_n) \in \X^n$, where $\X = \zo^{dk}$, and black-box access to mechanism $\mech$.
        \Statex\textbf{Output:}  $\vec{b} = (b_1,\dots,b_k) \in [d]^k$.
        \State Let $\vec{v}_j$ be a vector of length $dk$ with $d$ ones in coordinates $[dj]\setminus[d(j-1)]$
      and $0$ everywhere else; let $\overline{\vec{v}_j} \gets 1^{dk} - \vec{v}_j$.
        \State Construct a stream $\vec{\dstream} \gets \vec{\dset} \circ (0^{dk})^n \circ (\vec{v}_1)^{2n} \circ (\overline{\vec{v}_1})^{2n} \circ \dots \circ (\vec{v}_{k-1})^{2n} \circ (\overline{\vec{v}_{k-1}})^{2n} \circ (\vec{v}_k)^{2n}$ with $4kn$ records.\label{algline:construct-stream}
        \For{$t\in [T]$} 
        \State Send the record $\dstream_t$ to $\mech$ and get the 
        \ifnum\pods=0 
        corresponding
        \fi
        output $a_t$.
        \EndFor
        \For{$r\in[k]$}
        \State $b_r \gets a_{4rn} - d(r-1)$. If $b_r \not\in [d]$, then $b_r \gets 1.$
        \label{algline:kindsel_setoutput}
        \EndFor
        \State Output $\vec{b} \gets (b_1,\dots,b_k)$. 
    \end{algorithmic}
\end{algorithm}

\begin{lemma}\label{lem:kindsel-reduction}
Let $\alg$ be \Cref{alg:kindsel}. For all  $\eps > 0, \delta \geq 0$, $\alpha \in \R^{+}$, and $T,d,k,n \in \N$, where $T \geq 4kn$, if mechanism $\mech$ is $(\eps, \delta)$-differentially private and $(\alpha,T)$-accurate for $\selection{dk}$ in the \CR{}, then batch algorithm $\alg$ is $(\eps, \delta)$-differentially private and $(\frac \alpha n,n)$-accurate for $\kindselection{d}$.
\end{lemma}

\begin{proof}
We start by reasoning about privacy. Fix neighboring datasets $\vec{\dset}$ and $\vec{\dset'}$ that are inputs to algorithm~$\alg$. Let $\vec{\dstream}$ and $\vec{\dstream}'$ be the streams constructed in Step~\ref{algline:construct-stream} of $\alg$ when it is run on $\vec{\dset}$ and $\vec{\dset}'$, respectively. By construction, $\vec{\dstream}$ and $\vec{\dstream'}$ are neighboring streams. Since $\mech$ is $(\eps, \delta)$-DP, and $\alg$ only post-processes the outputs received from $\mech$,  Lemma~\ref{prelim:postprocess} implies that $\alg$ is $(\eps, \delta)$-DP.

Next, we reason about accuracy. Fix a dataset $\vec{\dset}$
and the corresponding data stream $\vec{\dstream}$ sent to  $\mech$. Consider a setting $\tau$ of the random coins of $\alg$. Since the transformation from $\mech$ to $\alg$ is deterministic, they correspond to coins used by $\mech$ when $\alg$ runs it as a subroutine.
Let $\alpha_{\tau}$ be the realized error of $\mech$ with coins $\tau$, that is,
    \begin{equation*}\label{eq:acctopsel1}
    \alpha_{\tau} = \max_{t \in [4kn]} \left( \err_{\selection{dk}}(\xt,a_t)\right),
    \end{equation*}
where  $a_t$ are the answers with coins $\tau$. Similarly, let $\gamma_{\tau}$ be the realized error of $\alg$ with coins $\tau$, that is,
 \begin{align*}\label{eq:acctopsel2}
    \gamma_{\tau} &= \err_{\kindselection{dk}}(\yvec,\vec{b}) \\
    &= \frac 1 n \cdot \max_{r \in [k]} \left( \err_{\selection{d}}(\yvec[(r-1)d+1:rd],b_r)\right),
    \end{align*}
    where $\vec{b} = (b_1,\dots,b_k)$ is the output of $\alg$ run with coins $\tau$.
    
The main observation in the accuracy analysis is that if $\alpha_{\tau}$ is small, so is $\gamma_{\tau}$. Note that if $\alpha\geq n$, the accuracy guarantee for $\alg$ is vacuous. Now assume $\alpha <n$.  For all blocks $r\in[k]$,
the sums in $\dstreamvec_{[4rn]}
=
\vec{\dset} \circ (0^{dk})^n \circ (\vec{v}_1)^{2n} \circ (\overline{\vec{v}_1})^{2n} \circ \dots \circ (\vec{v}_{r-1})^{2n} \circ (\overline{\vec{v}_{r-1}})^{2n} \circ (\vec{v}_r)^{2n} $
of all coordinates not in block $r$ are smaller than the sums of coordinates in block $r$ by at least $n$.
Consider coins $\tau$ with $\alpha_{\tau} \leq \alpha$. Since $\alpha_\tau<n,$ the index $a_{4rn}$ returned by $\mech$ is in block $r$ for all $r\in[k].$ Moreover, the error for each block is at most $\frac{\alpha_\tau}n$.
Therefore, $\gamma_\tau\leq\frac{\alpha_\tau}n\leq \frac \alpha n.$ Considering the probability of this event over all coins $\tau,$ we get
$$\Pr_{\text{coins $\tau$ of $\alg$}}\Big[\gamma_{\tau} \leq \frac \alpha n\Big] \geq
\Pr_{\text{coins $\tau$ of $\mech$}}[\gamma_{\tau} \leq \alpha] \geq \frac{2}{3},$$
where the last inequality holds because $\mech$ is $(\alpha,T)$-accurate.
We conclude that $\alg$ is $(\frac \alpha n,n)$-accurate.
\end{proof}

\ifnum\pods=0
Finally, we prove Theorem~\ref{thm:selection-LB}.
\fi
\begin{proof}[Proof of Theorem~\ref{thm:selection-LB}]

This proof's structure resembles that of Theorem~\ref{thm:main_maxsum}. First, for the case of ${\eps \leq \frac{2}{T}}$, we 
prove that $\alpha = \Omega(T)$.  Let $\vec{e}_j$ be a record of length $d$ with 
1 in coordinate $j$ and 0 everywhere else. Let $\vec{\dstream}=(\vec{e}_1)^{T/4}\circ (0^d)^{3T/4}$ and $\vec{\dstream}' =(\vec{e}_2)^{T/4}\circ (0^d)^{3T/4}$. Proceeding as in the proof of \Cref{thm:main_maxsum} (using group privacy and the error associated with selection) yields $\alpha = \Omega(T)$.

For all other values of $\eps$, we reduce from $\kindselection{}$, relying on the lower bounds for $\kindselection{}$ from \Cref{lem:kindsel-mainlb}.
Fix $T, d, \eps$. Given an integer $k$, the reduction of  \Cref{lem:kindsel-reduction} maps a batch instance of $\kindselection{d'}$ of size $n$ to an instance of $\selection{d}$ with $d = d'k$ and $T = 4nk$. The reduction applies as long as $d' = \frac{d}{k}\geq 2$ and $n= \frac{T}{4k}\geq 1$ are integers. We will ignore the integrality requirement (which can be addressed by appropriate padding) and allow any $k$ between $1$ and $\min(\frac d 2 , \frac T 4)$. 


%
When $\delta >0$, the reduction leads to a lower bound on the error of $\min\bparen{\Omega\bparen{\frac{\sqrt{k} \log d' }{\eps \log k}}, n} $ when $k$ and $d'$ are sufficiently large constants. In our setting, this translates to a lower bound of 
 $\Omega(\alpha_k)$ for  $\alpha_k = {\min\bparen{ \frac{\sqrt{k} \log(2 + d/k)}{\eps \log (2 + k)}, \frac{T}{k}}}$. 
(We add $2$ inside the logarithms to avoid 0 or subconstant log terms; this does not change the asymptotics.) Our goal is to select the value of $k \in [1, \min(\frac{d}{2},\frac{T}{4})]$ that maximizes $\alpha_k$.
For fixed $T, d, \eps$, let $k^* = k^*(T,d,\eps)= \max(1,k')$ where $k'$ denotes the largest value of $k$ where the two terms defining $\alpha_k$ equalize (that is, $k'$ satisfies $k'\sqrt{k'} \log(2 + d/k') / \log(2+ k') = \eps T$). 
We use two basic facts about $\alpha_k$: first, for $d> 15,000$, the function $\alpha_k$ is increasing on $[1,k^*)$ and decreasing on $(k^*, \infty)$. 
Second, its maximum value $\alpha_{k^*}$ is $\tilde \Omega(\frac{T^{1/3} \log^{2/3} d}{\eps ^{2/3}})$.

We consider four regimes for the triple $(T,d,\eps)$: 
\begin{enumerate}[label=(\emph{\alph*})]
    \item $k^*(T,d,\eps)=1$: In this case, $\alpha_k$ is maximized at $k=1$ and we obtain a lower bound of $\Omega(T/k) = \Omega(T)$.
    \item $k^*(T,d,\eps)>\min(\frac{d}{2},\frac{T}{4})$ and $2d \leq T$: In this case, we set $k=d/2$ and get a lower bound of $\alpha_k = \frac{\sqrt{k} \log(2 + d/k)}{\eps \log (2 + k)}$ (since $k \leq k^*$), which is $\Omega(\frac{\sqrt{d}}{\eps \log(2+d)}) = \tilde \Omega(\frac{\sqrt{d}}{\eps})$.
    \item  $k^*(T,d,\eps)>\min(\frac{d}{2},\frac{T}{4})$ and $2d > T$: This case is not possible for large $T$. For it to occur, we must have $k^* > T/4$, which implies that $\alpha_{k^*} < 4$. Since $\alpha_{k^*} = \tilde \Omega(\frac{T^{1/3} \log^{2/3} d}{\eps ^{2/3}})$, we get that $\eps>1$ (for sufficiently large $T$), contradicting our assumptions.
    \item $k^*(T,d,\eps) \in [1,\min(\frac{d}{2},\frac{T}{4})]$: In this case, we set $k=k^*$ and obtain a lower bound of $\alpha_{k^*} = \tilde \Omega(\frac{T^{1/3} \log^{2/3} d}{\eps ^{2/3}})$.
\end{enumerate}
Thus, for all possible relationships between $T,d$ and $\eps$, we obtain a lower bound that is one of three terms in the theorem statement.

The setting in which $\delta=0$ is similar. For a given $k \in [1,\min(\frac{d}{2},\frac{T}{4})]$, we obtain a lower bound of $\Omega(\alpha_k)$ for $\alpha_k = \min\big(\frac{k \log(2 + \frac d k)}{\eps}, \frac {T}{k} \big)$. The remaining calculations parallel the case where $\delta>0$, except that now $\alpha_{k^*}=\Theta\big(\sqrt{\frac{T\log d}{\eps}}\big)$.
\end{proof}

\section{Adaptive Upper Bounds}\label{sec:upper-bounds}
In this section, we define the adaptive continual release model and describe differentially private mechanisms for two types of problems in this model: $\selection{d}$ and approximating functions with bounded sensitivity ($\ell_2$ sensitivity in the case of approximate differential privacy and $\ell_1$ sensitivity in the case of pure DP).  
Our mechanisms are $(\alpha,T)$-accurate, where the upper bounds for $\alpha$ match the lower bounds obtained in previous sections in the \CR{} up to logarithmic factors in the time horizon $T$, the number of coordinates $d$, and the inverse of the privacy parameter $\frac{1}{\delta}$.

\ifnum\pods=1
\begin{table*}[t]
\begin{tabular}{|c||c|c|c|c|}
\hline  
   & \multicolumn{2}{c|}{$\rho$-zCDP}  & \multicolumn{2}{c|}{$(\eps,0)$-DP}  \\
   \cline{2-5}
     Problem      & Tree & Recomputation & Tree  & Recomputation  \\
     \hline \hline 
     \maxsum{} 
     & \( \frac{\sqrt{d} \log T \sqrt{\log(dT)}}{\sqrt{\rho}} \)
     & \( \sqrt[3]{\frac{T\log T}{\rho}} \)
     & \( \frac{d (\log d) \log^3 T}{\eps} \)
     & \(\sqrt{\frac{T\log T}{\eps}} \)
     \\
     \hline 
     \selection{} 
     & \( \frac{\sqrt{d} \log T \sqrt{\log(dT)}}{\sqrt{\rho}} \)
     & \( \frac{T^{1/3}\log^{2/3}(d T)  }{\rho^{1/3}} \)
     & \(  \frac{d (\log d) \log^3 T}{\eps} \)
     & \(\sqrt{\frac{T\log (dT)}{\eps}} \)
     \\
     \hline 
\end{tabular}
\caption{Our (asymptotic) upper bounds on the error of DP mechanisms in the adaptive continual release model.}    
\label{tab:algorithmic-results}
\end{table*}
\else
\fi


\subsection{Adaptive Continual Release}
\label{sec:defs-adaptive}

In the \aCR{}, the input stream given to a mechanism $\mech$ is chosen adversarially. That is, $\mech$ interacts with a randomized adversarial process $\adv$ that runs for $T$ timesteps; at timestep $t \in [T]$, the process $\adv$ receives $a_t$ from $\mech$, updates its internal state, and produces input record $\random{\dstream}_{t+1}$ that is sent to $\mech$ at timestep $t+1$. Process $\adv$ can choose $\dstream_{t+1}$ based on the previous input records $\xt$ and $\mech$'s previous outputs $\vec{a}_{[t]}$. We make no assumptions on $\adv$ regarding running time or complexity; its only limitation is that it does not see the internal coins of $\mech$.

\begin{definition}
A mechanism $\mech$ is  \emph{($\alpha$,T)-accurate} for a function $f$ in the \aCR{} if for all processes $\adv$, the error of $\mech$ with respect to $\adv$ is at most $\alpha$ with high probability, that is,
$$\Pr_{\text{coins of } \mech,\adv} \left[\max_{t\in [T]} \err_f(a_t; \xt) \leq \alpha\right] \geq \frac{2}{3}.$$
\end{definition}

A similar notion of accuracy was considered in work on adversarial streaming \cite{EliezerJWY20, HassidimKMMS20, KaplanMNS21}, though those articles do not directly address privacy.

Next, we define \emph{(event-level) privacy in the \aCR{}} \ifnum\pods=0
, which is trickier than in the \CR{}
\else
\fi
. This concept is implicit in \cite{SmithT13}, but to our knowledge has not been previously defined. Privacy is defined with respect to the game $\Pi_{\mech,\adv}$, described in \Cref{alg:privacy_game}, between mechanism $\mech$ and an adversary $\adv$. In all timesteps except one, $\adv$ outputs a single input record which $\Pi_{\mech,\adv}$ simply forwards to $\mech$. However, there is a special {\em challenge timestep} $t^* \in [T]$, selected by $\adv$, in which $\adv$  provides two records $x_{t^*}^{(L)}$ and $x_{t^*}^{(R)}$. 
The game comes in two versions, specified by its input parameter $\side \in \{L,R\}$ which is not known to $\adv$ or $\mech$: in one version, the record $x_{t^*}^{(L)}$ is handed to $\mech$ at timestep $t^*$; in the other, the record $x_{t^*}^{(R)}$ is handed to $\mech$ instead. The mechanism is private if the distributions on the adversary's view, which consists of its internal randomness and the transcript of messages it sends and receives, 
are close in the two versions of the game.

When the adversary decides in advance on all $T+1$ records that it outputs over the course of the game, the resulting definition is equivalent to the nonadaptive version (Definition~\ref{def:na-privacy}). The version we give here captures a richer class of settings. 

Intuitively, we may think of $x_{t^*}^{(L)}$ as the data of person $t^*$, and of $x_{t^*}^{(R)}$ as a dummy value (say, all 0's). The parameter $\side$ then controls whether the data of person $t^*$ is included in the computation or not. The privacy requirement is that an outside attacker cannot tell whether $t^*$'s data was used, even if the attacker has full knowledge of the process generating the data stream. The adversary $\adv$ combines the data generation process and the attack itself in one entity, so that our model allows for an arbitrary relationship between them.

\newcommand{\rdtype}{\text{\sf type}}
\newcommand{\chall}{\text{\sf challenge}}
\newcommand{\reg}{\text{\sf regular}}

\begin{algorithm}[ht]
\caption{Privacy game $\Pi_{\mech,\adv}$ 
\ifnum\pods=0
for the \aCR{}
\fi
}
\label{alg:privacy_game}
    \begin{algorithmic}[1]
        \Statex \textbf{Input:}  time horizon $T \in \N$, $\side \in \{L, R\}$ (not known to $\adv$).
        \For{\text{$t=1$ to $T$}}
        \State $\adv$ outputs $\rdtype_t \in \{\chall, \reg\}$, where $\chall$ is chosen once during the game.
            \If{$\rdtype_t = \reg$}
                \State $\adv$ outputs $x_t \in \X$ which is sent to $\mech$.
            \EndIf
            \If{$\rdtype_t = \chall$}
                \State $t^* \gets t$. 
                \State $\adv$ outputs $(x_t^{(L)},x_t^{(R)}) \in \X^2$.
                \State $x_t^{(\side)}$ is sent to $\mech$.  
            \EndIf
            \State $\mech$ outputs $a_t$ which is given to $\adv$.
        \EndFor
    \end{algorithmic}
\end{algorithm}

\begin{definition}
The \emph{view of $\adv$} in privacy game $\Pi_{\mech,\adv}$  consists of $\adv$'s internal randomness and the transcript of messages it sends and receives. Let \emph{$V_{\mech,\adv}^{(\side)}$} denote $\adv$'s  view  at the end of the game run with input $\side \in \{L,R\}$.
\end{definition}

One could also define the adversary's view as its internal state at the end of the game. The version we define contains enough information to compute that internal state, but is simpler to work with. 

In addition to $(\eps,\delta)$-DP, we consider a related notion, called zCDP \cite{BunS16}. See Appendix~\ref{sec:zCDP} for background on zCDP and the notion of $\rho$-closeness of random variables ($\simeq_{\rho})$.

\begin{definition}
A mechanism $\mech$ is \emph{$(\epsilon,\delta)$-DP in the \aCR{}} if, for all adversaries~$\adv$, 
$$V_{\mech,\adv}^{(L)} \approx_{\eps,\delta} V_{\mech,\adv}^{(R)}. $$

\noindent A mechanism $\mech$ is \emph{$\rho$-zCDP in the \aCR{}} if for all adversaries~$\adv$,
\begin{equation*} 
V_{\mech,\adv}^{(L)}\simeq_{\rho} V_{\mech,\adv}^{(R)}. 
\end{equation*}
\sstext{The symbol $\simeq_{\rho}$ denotes $\rho$-closeness (\Cref{def:rho-indistinguishable})}.
\end{definition}


\ifnum\pods=1
\subsection{Summary of Adaptive Upper Bounds}\label{sec:adaptive-up-summary}
Our upper bounds on the error of differentially private algorithms for $\maxsum{d}$ and $\selection{d}$ in the adaptive continual release model are summarized in Table~\ref{tab:algorithmic-results}. The corresponding theorems are stated in Section~\ref{sec:adaptive-up-statments}. The upper bounds in the table are attained by two simple mechanisms: one uses the binary tree mechanism and the other recomputes the target function at regular intervals. The bounds stated for $\maxsum{d}$ and obtained via recomputing periodically apply more generally: to all sensitivity-1 functions.

Unlike the proofs of privacy in previous work, the privacy proofs for our adaptive upper bounds (in \Cref{sec:details_ub}) use a simulator that constructs inputs to an underlying mechanism online,  instead of considering only neighboring datasets defined ahead of time.
This allows us to reason about input streams that look very different in different executions of the game (for example, streams whose long-term behavior depends on the mechanism's response at the challenge timestep). 
The accuracy analyses need to account for adaptivity as well.

\else
\fi

\ifnum \pods=1
\addcontentsline{toc}{section}{References}
\bibliographystyle{plain}
\bibliography{bibliography}
\appendix
\else
\fi

\ifnum\pods=0
\subsection{Statements of Adaptive Upper Bounds}\label{sec:adaptive-up-statments}
\else
\section{Additional Preliminaries}\label{app:preliminaries}

\subsection{Facts on Differential Privacy}
Differential privacy protects groups of individuals.

\begin{lemma}[Group Privacy~\cite{dwork2006calibrating}]\label{lem:group_privacy} Every $(\eps, \delta)$-DP algorithm \alg is $\left(\ell \eps, \delta' \right)$-DP for groups of size $\ell$, where $\delta' = \delta\frac{e^{\ell \eps} -1}{e^\eps-1}$; that is, for all datasets $\vec{\dstream}, \vec{\dstream}'$ such that $\|\vec{\dstream} - \vec{\dstream}' \|_0 \leq \ell$,
\begin{equation*}
    \alg(\vec{\dstream}) \approx_{\ell \eps, \delta'}   \alg(\vec{\dstream}').
\end{equation*}
\end{lemma}

Differential privacy is closed under post-processing.
\begin{lemma}[Post-Processing~\cite{dwork2006calibrating,BunS16}]\label{prelim:postprocess} If $\alg$  is an $(\eps, \delta)$-DP algorithm with output space $\mathcal{Y}$ and $\mathcal{B}$  is a randomized map from $\mathcal{Y}$ to $\mathcal{Z}$, then the algorithm $\mathcal{B} \circ \alg$ is $(\eps, \delta)$-DP.
\end{lemma}

\begin{definition}[Sensitivity] Let $f: \X^n \rightarrow \R^m$ be a function. Its $\ell_1$-sensitivity is
\begin{equation*}
    \sup_{\text{neighbors } \vec{\dstream}, \vec{\dstream}' \in \X^n} \|f(\vec{\dstream}) - f(\vec{\dstream}')\|_1.
\end{equation*}
To define $\ell_2$-sensitivity, we replace the $\ell_1$ norm with the $\ell_2$ norm.
\end{definition}

We use the standard Laplace and exponential mechanisms to ensure differential privacy. 

\begin{definition}[Laplace Distribution] The Laplace distribution with parameter $b$ and mean $0$, denoted  $\Lap(b)$,
has probability density
\begin{equation*}
    h(r) = \frac{1}{2b}e^{-\frac{|r|}{b}} \text{ for all $r \in \mathbb{R}$}.
\end{equation*}
\end{definition}

\begin{lemma}[Laplace Mechanism, \cite{dwork2006calibrating}]\label{prelim:laplace_dp} Let $f : \X^n \rightarrow \mathbb{R}^m$ be a function with $\ell_1$-sensitivity at most $\Delta_1$. Then the Laplace mechanism is algorithm
\begin{equation*}
    \alg_f(\vec{\dstream}) = f(\vec{\dstream}) + (Z_1, \ldots, Z_m),
\end{equation*}
where $Z_i \sim \Lap\left(\frac{\Delta_1}{\eps}\right)$. Algorithm $\alg_f$ is $(\eps, 0)$-DP.
\end{lemma}

\begin{lemma}[Exponential Mechanism \cite{McTalwar}]\label{lem:expmech}
Let $L$ be a set of outputs and $g: L \times \X^n \to \mathbb{R}$ be a function that measures the quality of each output on a dataset. Assume that for every $m \in L$, the function $g(m,.)$ has $\ell_1$-sensitivity at most $\Delta$. Then, for all $\eps,n > 0$ and  all datasets $\dset \in \X^n$, there exists an $(\eps, 0)$-DP mechanism that outputs an element $m\in L$ such that, for all $a>0$,
\begin{equation*}
    \Pr[\max_{i \in [L]} g(i,\dset) -  g(m,\dset) \geq 2\Delta \frac{(\ln |L| + a)}{\eps}] \leq e^{-a}. 
\end{equation*}
\end{lemma}

\subsection{Facts on $\rho$-zCDP}\label{sec:zCDP}

This section contains facts about {\em zero-concentrated differential privacy} (zCDP). It differs from $(\eps,\delta)$-differential privacy in that zCDP requires output distributions on all pairs of neighboring datasets to be $\rho$-close (\Cref{def:rho-indistinguishable}) instead of $(\eps,\delta)$-indistinguishable.
In \Cref{sec:upper-bounds}, we prove that our mechanisms are zCDP and then use conversion to $(\eps,\delta)$-differential privacy (\Cref{prop:CDPtoDP}) to compare our upper and lower bounds.

\begin{definition}[R\'enyi Divergence \cite{Renyi61}]
Let $Q$ and $Q'$ be distributions on $\mathcal{Y}$. For $\xi \in (1,\infty)$, the R\'enyi divergence of order $\xi$ between $Q$ and $Q'$(also called the $\xi$-R\'enyi Divergence) is defined as
\begin{align}
    D_{\xi}(Q \| Q') = \frac{1}{\xi-1} \log\left( \E_{r \sim Q'} \left[ \left(\frac{Q(r)}{Q'(r)}\right)^{\xi-1} \right]  \right).
\end{align}
Here $Q(\cdot)$ and $Q'(\cdot)$ denote either probability masses (in the discrete case) or probability densities (when they exist). More generally, one can replace  $\frac{Q(.)}{Q'(.)}$ with the the Radon-Nikodym derivative of $Q$ with respect to $Q'$.
\end{definition}

\begin{definition}[$\rho$-Closeness]\label{def:rho-indistinguishable}
Random variables $R_1$ and $R_2$ over the same outcome space $\mathcal{Y}$ are  {\em $\rho$-close} (denoted $R_1 \simeq_{\rho} R_2$) if for all $\xi \in (1,\infty)$, 
\begin{align*}
D_{\xi}(R_1\|R_2) \leq \xi\rho \text{ and }  D_{\xi}(R_2\|R_1) \leq \xi\rho,
\end{align*}
where $D_{\xi}(R_1\|R_2)$ is the $\xi$-R\'enyi divergence  between the distributions of $R_1$ and $R_2$.
\end{definition}

\begin{definition}[zCDP in Batch Model~\cite{BunS16}]
A randomized batch algorithm $\alg : \X^n \to \mathcal{Y}$ is $\rho$-zero-concentrated differentially private ($\rho$-zCDP), if, for all neighboring datasets $\vec{\dset},\vec{\dset}' \in \mathcal{X}^n$,
$$\alg(\vec{\dset}) \simeq_{\rho} \alg(\vec{\dset'}).$$
\end{definition}

One major benefit of using zCDP is that this definition of privacy admits a clean composition result. 

\begin{lemma}[Composition \cite{BunS16}] \label{lem:cdp_composition}
Let $\alg : \mathcal{X}^n \to \mathcal{Y}$ and $\alg' : \mathcal{X}^n \times \mathcal{Y} \to \mathcal{Z}$ be batch algorithms. Suppose $\alg$ is $\rho$-zCDP and $\alg'$ is $\rho'$-zCDP. Define batch algorithm $\alg'' : \mathcal{X}^n \to \mathcal{Y} \times \mathcal{Z}$ by $\alg''(\vec{\dset}) = \alg'(\yvec,\alg(\yvec))$. Then $\alg''$ is $(\rho+\rho')$-zCDP.
\end{lemma}

The \emph{Gaussian mechanism}, defined next,
privately estimates a real-valued function by adding Gaussian noise to its value. 

\begin{definition}[Gaussian Distribution] The Gaussian distribution with parameter $\sigma$ and mean 0, denoted $\Gauss(0,\sigma^2)$, has probability density
\begin{equation*}
    h(r) = \frac{1}{\sigma \sqrt{2\pi}} e^{-\frac{r^2}{2\sigma^2}}  \text{ for all $r \in \mathbb{R}$}.
\end{equation*}
\end{definition}

\begin{lemma}[Gaussian Mechanism \cite{BunS16}] \label{prop:gaussian-mech}
Let $f : \X^n \to \mathbb{R}$ be a function with $\ell_2$-sensitivity at most $\Delta_2$. Let $\alg$ be the batch algorithm 
that, on input $\yvec$, releases a sample from $\mathcal{N}(f(\yvec), \sigma^2)$. Then $\alg$ is $(\Delta_2^2/2\sigma^2)$-zCDP.
\end{lemma}

The final lemma in this section relates zero-concentrated differential privacy to $(\eps,\delta)$-differential privacy.

\begin{lemma}[Conversion from zCDP to DP \cite{BunS16}]\label{prop:CDPtoDP}
For all $\rho,\delta > 0$, if batch algorithm $\alg$ is $\rho$-zCDP, then $\alg$ is $(\rho+2\sqrt{\rho \log(1/\delta)},\delta)$-DP.
\end{lemma}

\subsection{Useful Concentration Inequalities}
\begin{lemma}\label{lem:gauss_conc}
For all random variables $R \sim \gauss(0,\sigma^2)$,
\begin{equation*}
    \Pr[|R| > \ell] \leq 2e^{-\frac{\ell^2}{2 \sigma^2}}.
\end{equation*}
\end{lemma}
\begin{lemma}\label{lem:gauss_max}
Consider $m$ random variables $R_1,\dots,R_m \sim \gauss(0,\sigma^2)$. Then
\begin{equation*}
    \Pr[\max_{j \in [m]} |R_j| > \ell] \leq 2 m e^{-\frac{\ell^2}{2 \sigma^2}}.
\end{equation*}
\end{lemma}
\begin{proof}
By a union bound and Lemma~\ref{lem:gauss_conc},
\begin{align*}
    \Pr[\max_{i \in [m]} |R_i| > \ell] 
      = \Pr( \exists i \in [m] \text{ such that } |R_i| > \ell) \\
     \leq \sum_{i=1}^m \Pr( |R_i| > \ell) \leq
    \sum_{i=1}^m  2e^{-\frac{\ell^2}{2 \sigma^2}}  =  2m e^{-\frac{\ell^2}{2 \sigma^2}}.
\end{align*}
\end{proof}
A similar union bound argument yields the following concentration inequality on the maximum of the absolute values of i.i.d.\ Laplace random variables.
\begin{lemma}\label{lem:lapmax}
Fix $m \in \N$, $\lambda > 0$. Consider $m$ random variables $R_1,\dots,R_m \sim Lap(\lambda)$. Then for all $a>0$,
$$\Pr(\max_{i \in [m]} |R_i| > \lambda (\log m + \log a) ) \leq e^{-a}.$$
\end{lemma}

\section{Proofs  Omitted from Section~\ref{sec:selection-lowerbound}}\label{app:sec4proofs}

Item~1 in Lemma~\ref{lem:kindsel-mainlb} follows from \Cref{lem:lbkindsel} below.
\begin{theorem}[\cite{US17,Ullman21pers}]\label{lem:lbkindsel}
For all $\epsilon \in (0,1], \delta \in (0,1/n]$, $\gamma \in [0,\frac{1}{20}]$, $d,n,k \in \N$, if Algorithm $\alg$ is $(\eps,\delta)$-differentially private and $(\gamma,n)$-accurate for $\kindselection{d}$, then $n = \Omega(\frac{\sqrt{k}\log d}{\gamma \eps \log (k+1)})$.
\end{theorem}
\begin{proof}[Proof Sketch]
We are aware of two proofs of this result, both of which were communicated to us by Jonathan Ullman~\cite{Ullman21pers}. The first uses the top-$k$ selection lower bound of Steinke and Ullman~\cite{US17}. In that problem, there is a single collection of $d$ coordinates and the goal is to return the indices of $k<d$ coordinates whose sums are roughly largest. 

For the specific distribution over instances that arises in the lower bound of~\cite{US17}, if one divides the coordinates into $k$ equal groups, there is a constant probability that the collection of coordinates with the largest sum in each group is a good approximate solution for the top-$k$ selection problem. An algorithm for $\kindselection{d}$ can thus be used to solve the top-$k$ selection (out of $dk$ coordinates) problem for such instances with roughly the same error and privacy parameter. The lower bound of~\cite{US17} on $n$ then applies.

Another approach is to use the composition framework of Bun, Ullman and Vadhan \cite{bunUV18}. One can use a folklore result that selection among $d>2^m$ coordinates can be used to mount a reconstruction attack on an appropriate dataset of size $m$. Composed with the lower bound for 1-way marginals in \cite{bunUV18}, one obtains a lower bound for $\kindselection{d}$.
\end{proof}
To complete the proof of \Cref{lem:kindsel-mainlb}, we prove Item $2$ via a standard packing argument. 
\begin{proof}[Proof of \Cref{lem:kindsel-mainlb},  Item~2.]
For $\vec{u}\in [d]^k,$ let $\dset^*_{\vec{u}} \in \zo^{dk}$ be the record where each block $r\in[k]$ of $d$ coordinates has a $1$ in coordinate $u_r$ and zeros everywhere else. Let $\vec{\dset_u}$ be the dataset that consists of $2\gamma n$ copies of~$\dset^*_{\vec{u}}$ and $(1-2\gamma)n$ copies of the all-zero record (assume, for simplicity, that $2\gamma n$ is an integer). Since $\alg$ is $(\gamma,n)$-accurate,\\
$\Pr_{\text{coins of \alg}}\left[ \err_{\kindselection{}}(\vec{\dset_u},\alg(\vec{\dset_u})) \leq \gamma \right] \geq \frac{2}{3}$ for all $\vec{u}\in[d]^k$. This means that for all $\vec{u}\in[d]^k$,
\[\Pr_{\text{coins of \alg}}\left[\alg(\vec{\dset_u}) = \vec{u}\right] \geq \frac{2}{3}.\]
For all $\vec{u},\vec{u}' \in [d]^k$, by group privacy,
$\alg(\vec{\dset_u}) \approx_{(\gamma\eps n,0)} \alg(\vec{\dset_{u'}})$, which implies that
\begin{align}\label{eq:groupprivlb}
\Pr\left[\alg(\vec{\dset_{u}}) = \vec{u}'\right] &\geq e^{-\gamma\eps n}\Pr\left[\alg(\vec{\dset_{u'}}) = \vec{u}'\right] \geq \frac{2}{3}e^{-\gamma\eps n}.
\end{align}
Since the probability of any event is at most 1,
\begin{align*}
& 1\geq \Pr_{\text{coins of \alg}}\left[ \alg(\vec{\dset_u}) \neq \vec{u} \right] = \\
& \sum_{\vec{u}'\neq \vec{u}}\Pr\left[\alg(\vec{\dset_u}) = \vec{u}'\right] \geq \frac{2}{3}e^{-\gamma\eps n}(d^k-1),
\end{align*}
where the last inequality holds by (\ref{eq:groupprivlb}). We get that
$e^{\gamma\eps n} \geq \frac{d^k-1}{2}\cdot\frac{2}{3}$, and thus $n = \Omega\left(\frac{k\log d}{\gamma \cdot \eps}\right).$
\end{proof}

\section{Details Omitted from Section~\ref{sec:upper-bounds}}\label{sec:details_ub}
\subsection{Formal Statements}\label{sec:adaptive-up-statments}
\fi
In this subsection, we state theorems that summarize the performance guarantees of our mechanisms for $\maxsum{d}$ and $\selection{d}$. We prove these theorems in the following subsections. The upper bounds in these theorems are attained by two simple mechanisms: one uses the binary tree mechanism and the other recomputes the target function at regular intervals. We first state results for the binary-tree-based approach.

\begin{theorem}[zCDP, Binary-Tree-Based Mechanisms]\label{thm:bintree_rho}
For all $\rho \in (0,1]$, $d \in \N$, and sufficiently large $T > 0$, there exist $\rho$-zCDP mechanisms $\mech, \mech'$ in the \aCR{} such that $\mech$ is $(\alpha,T)$-accurate for $\maxsum{d}$ and $\mech'$ is $(\alpha,T)$-accurate for $\selection{d}$, where $\alpha = \BigO{\frac{\sqrt{d} \log T \sqrt{\log(dT)}}{\sqrt{\rho}}}$.
\end{theorem}

The next theorem uses the idea of recomputing at regular intervals, which applies quite generally. Item 1 of Theorem~\ref{thm:recompute_rho} applies for general sensitivity-1 functions (which include $\maxsum{d}$); a similar result holds for bounded-sensitivity functions with output space $\R^d$. 

\begin{theorem}[zCDP, Mechanisms via Recomputing at Regular Intervals]\label{thm:recompute_rho}
For all $\rho \in (0,1]$, $d \in \N$, sufficiently large $T > 0$, and all functions $f:\X^* \to \R$ with $\ell_2$-sensitivity at most $1$, there exist $\rho$-zCDP mechanisms $\mech$ and $\mech'$ in the \aCR{} such that

\begin{enumerate}
    \item Mechanism $\mech$ is $(\alpha,T)$-accurate for $f$ for \ifnum \pods=1 \\ \fi $\alpha \allowbreak = \BigO{\min\left\{ \sqrt[3]{\frac{T\log T}{\rho}},T \right\}};$
    \item Mechanism $\mech'$ is $(\alpha,T)$-accurate for $\selection{d}$ for $\alpha = \BigO{\min\left\{ \frac{T^{1/3}\log^{2/3}(d T)  }{\rho^{1/3}}, T \right\} }$.
\end{enumerate} 
\end{theorem}

Combining Theorems~\ref{thm:bintree_rho}--\ref{thm:recompute_rho}, using the conversion from zCDP to $(\eps, \delta)$-DP from \Cref{prop:CDPtoDP} and substituting $\rho = \frac{\eps^2}{16 \log(1/\delta)}$, we get the following corollary.
\begin{corollary}\label{cor:upper-bounds}
For all $\eps \in (0,1]$, $\delta \in (0,\frac{1}{2}]$, $d \in \N$, and sufficiently large $T > 0$, there exist $(\eps, \delta)$-DP mechanisms $\mech$ and $\mech'$ in the \aCR{} such that\\
(1) $\mech$ is $(\alpha,T)$-accurate for $\maxsum{d}$ for 
\ifnum\pods=1
\\
\else
\fi
$\alpha = \BigO{\min\left\{ \frac{\sqrt[3]{T\log(1/\delta)\log T} }{\eps^{2/3}}, \frac{\sqrt{d\log(dT) \log(1/\delta)}\,\log T}{\eps}, T \right\}};$\\
(2) $\mech'$ is $(\alpha,T)$-accurate for $\selection{d}$ for
\ifnum\pods=1
\\
\else
\fi
$\alpha = \BigO{ \min \left\{\frac{\sqrt[3]{T \log^2 (dT) \log(1/\delta)}  }{\eps^{2/3}},  \frac{ \sqrt{d\log(dT) \log(1/\delta)} \, \log T}{\eps} , T \right\}}$. 
\end{corollary}

Simple variants of our mechanisms can be used to get the following theorems for $(\eps,0)$-differential privacy. 

\begin{theorem}[Pure DP, Binary-Tree-Based Mechanisms]\label{thm:bintree_pure}
For all $\eps \in (0,1]$, $d \in \N$, and sufficiently large $T > 0$, there exist $(\eps,0)$-DP mechanisms $\mech$ and $\mech'$ in the \aCR{} such that $\mech$ is $(\alpha,T)$-accurate for $\maxsum{d}$ and $\mech'$ is $(\alpha,T)$-accurate for $\selection{d}$ for $\alpha = \BigO{\frac{d (\log d) \log^3 T}{\eps}}$.
\end{theorem}

\begin{theorem}[Pure DP, Mechanisms via Recomputing at Regular Intervals]\label{thm:recompute_pure}
For all $\eps \in (0,1]$, $d \in \N$, sufficiently large $T > 0$, and all functions $f:\X^* \to \R$ with $\ell_1$-sensitivity at most $1$, there exist $(\eps,0)$-DP mechanisms $\mech$ and $\mech'$ in the \aCR{} such that
\begin{enumerate}
    \item Mechanism $\mech$ is $(\alpha,T)$-accurate for $f$ for
\ifnum\pods=1
\\
\else
\fi
    $\alpha = \BigO{\min\left\{\sqrt{\frac{T\log T}{\eps}},\,T \ \right\}};$
    \item Mechanism $\mech'$ is $(\alpha,T)$-accurate for $\selection{d}$ for 
\ifnum\pods=1
\\
\else
\fi
$\alpha = \BigO{\min\left\{\sqrt{\frac{T \log (dT)  }{\eps}},\,T \right\} }.$
\end{enumerate}
\end{theorem}
Theorems~\ref{thm:bintree_pure}--\ref{thm:recompute_pure} yield the following corollary.

\begin{corollary}\label{cor:pure_dp}
For all $\eps \in (0,1]$, $d \in \N$, and sufficiently large $T > 0$, there exist $(\eps, 0)$-DP mechanisms $\mech$ and $\mech'$ in the \aCR{} such that
\begin{enumerate}
    \item  $\mech$ is $(\alpha,T)$-accurate for $\maxsum{d}$ for
\ifnum\pods=1
\\
\else
\fi
$\alpha = \BigO{\min\left\{\sqrt{\frac{T\log T}{\eps}},\,T, \frac{d (\log d) \log^3 T}{\eps} \ \right\}};$\\
    \item Mechanism $\mech'$ is $(\alpha,T)$-accurate for $\selection{d}$ for $\alpha = \BigO{\min\left\{\sqrt{\frac{T \log (dT)  }{\eps}},\,T,\frac{d (\log d) \log^3 T}{\eps} \right\} }$.
\end{enumerate}
\end{corollary}

\subsection{Algorithms based on the Binary Tree Mechanism}

In this section, we prove \Cref{thm:bintree_rho} for $\selection{d}$. \Cref{thm:bintree_rho} for $\maxsum{d}$ follows from the same analysis by considering the binary tree mechanism that outputs the highest noisy sum instead of the coordinate that achieves it.

In order to approximate $\selection{d}$ on a dataset with $d$ attributes, we use the binary tree mechanism from \cite{ChanSS10,DworkNPR10} to privately sum each of the attributes of the records $\xt$ received so far, and then choose the attribute with the highest sum. For simplicity of exposition, in this section, we assume that $T$ is a power of 2. In general, we can work with the smallest power of 2 greater than $T$. Throughout this section, $[i:j]$, where $i,j\in \mathbb{N}$, denotes the set of natural numbers $\{i, \dots, j\}$.

\newcommand{\nleft}{\ell}
\newcommand{\nright}{r}

\ifnum\pods=0
\begin{algorithm}[h!]
\else
\begin{algorithm*}[h!]
\fi
    \caption{Mechanism $\mech$ for $\selection{d}$ in \aCR{}}
    \label{alg:selbintree}
    \hspace*{\algorithmicindent}
    \begin{algorithmic}[1]
           \Statex\textbf{Input:} time horizon $T \in \N$, privacy parameter $\rho$, stream $\vec{\dstream} = (\dstream_1,\dots,\dstream_T) \in \X^T$, where $\X = \{0,1\}^d$.
           \Statex\textbf{Output:} stream $(a_1,\dots,a_T) \in [d]^T$.
           \State\textbf{Initialization:} Construct a complete binary tree with $T$ leaves labeled $\node_{[1:1]},\dots,\node_{[T:T]}$. Label every internal node  $\node_{[\nleft:\nright]}$
            if the subtree rooted at that node has leaves  $\node_{[\nleft:\nleft]},\dots,\node_{[\nright:\nright]}.$
            Initialize the partial sum $\nodeval_{[\nleft:\nright]}\gets 0^d$ for each node $\node_{[\nleft:\nright]}$ in the tree.
            \For{\text{$t=1$ to $T$}}
               \State Get record $\dstream_t$ from $\adv$. 
               
               \emph{$\triangleright$ Compute noisy sums for nodes completed at time $t$}
               \For{\text{each node $\node_{[\nleft:t]}$}}
               \State Draw noise $Z \sim \gauss(0,\sigma^2 \mathbb{I}^{d \times d})$, where $\sigma = \sqrt{\frac{d(\log T+1)}{2\rho}}$, and set $\nodeval_{[\nleft:t]} \gets \sum_{i=\nleft}^{\nright} \dstream_i +Z.$
               \label{line:noisebintree}
               \EndFor
               
               \emph{$\triangleright$ Output Steps:}
               
               \ifnum\pods=0
               \State  \label{step:output-start} $I_t \gets $ 
               \parbox[t]{5.5in}{collection  of at most $\log t + 1$ disjoint intervals whose union is $[1:t]$ and where  each interval labels a node in the binary tree. (See Remark~\ref{rem:bintree_dyadic}.)}
               \else
               \State  \label{step:output-start} $I_t \gets $ 
               \parbox[t]{5.5in}{collection  of at most $\log t + 1$ intervals that partition $[1:t]$, where  each interval labels a node in the binary tree. (See Remark~\ref{rem:bintree_dyadic}.)}
               \fi
               \State $\mathsf{sum}_t \gets \sum_{[\nleft:\nright] \in I_t}  \nodeval_{[\nleft:\nright]}$. \label{step:noisysumvec}
               \State  Output $a_t \gets \argmax_{j \in [d]} \mathsf{sum}_t[j]$. \label{step:a_tchoice} \label{line:outputselect}
               \EndFor
    \end{algorithmic}
\ifnum\pods=1
\end{algorithm*}
\else
\end{algorithm}
\fi
At the high level, the binary tree mechanism constructs a complete binary tree with $T$ leaves. The leaves correspond to the input records 
$\xt$, where each record $\dstream_i\in\zo^d$. Each internal node in the tree corresponds to the sum of all the leaves in its subtree. 
Each node stores the noisy version of the corresponding sum 
computed by adding a noise vector drawn from $\gauss(0,\sigma^2 \mathbb{I}^{d \times d})$ with $\sigma = \sqrt{\frac{d (\log T+1)}{2\rho}}.$ The algorithm that releases the noisy sum is $\frac{\rho}{\log T+1}$-zCDP.
Since each $\dstream_t$ participates in only $\log_2 T +1$ sums in the tree, by \emph{adaptive composition} of zCDP (\Cref{lem:cdp_composition}), the complete mechanism is $\rho$-zCDP (\Cref{thm:bintree_rho}). The sum of all the attributes at any timestep can be calculated by adding at most $\log T$ of the sums stored in the tree, one at each level. The algorithm that adds the corresponding noisy sums is $(\alpha,T)$-accurate for $\alpha\approx\BigO{ \frac{\sqrt{d} \log T \log(Td)}{\sqrt{\rho}}}$. The formal description of the algorithm appears in \Cref{alg:selbintree}. The algorithm uses a dyadic decomposition (described in Remark~\ref{rem:bintree_dyadic}) to decide which nodes of the tree it accesses to compute any particular output.

\begin{remark}[Dyadic Decomposition]\label{rem:bintree_dyadic}
For any natural number $t>1$, the interval $[1:t]$ can be expressed as a union of at most $\log t + 1$ disjoint intervals as follows. Consider the binary expansion of $t$ (which has at most $\log t + 1$ bits), and express $t$ as a sum of distinct powers of $2$ ordered from higher to lower powers. Then, the first interval $[1:r]$ will have size equal to the largest power of $2$ in the sum. The second interval will start at $r+1$ and its size will be equal to the second largest power of $2$ in the sum. Similarly, the remaining intervals are defined until all terms in the summation have been exhausted. For example, for $t=7=4+2+1$, the intervals are $[1:4]$, $[5:6]$ and $\{7\}$.
\end{remark}

We  present the privacy and accuracy analysis for \Cref{alg:selbintree} in Lemmas~\ref{lem:privbintree} and~\ref{lem:accbintree}, respectively, which together prove \Cref{thm:bintree_rho} for $\selection{}$.

\begin{lemma}\label{lem:privbintree}
For all $\rho \in \R^{+}$, $d,T \in \N$, mechanism $\mech$ described in \Cref{alg:selbintree} is $\rho$-zCDP in the \aCR.
\end{lemma}

\ifnum\pods=1
\begin{algorithm*}[h!]
\else
\begin{algorithm}[h!]
\fi
    \caption{Simulator $\Sim$ for the proof of \Cref{lem:privbintree}}
    \label{alg:binary_sim}
    \begin{algorithmic}[1]
        \Statex \textbf{Input:} \Longunderstack[l]{time horizon $T \in \N$, privacy parameter $\rho\in\R^{+}$, black-box access to an adversary $\adv$ and \\ a mechanism $\idealgauss$.}
        \Statex \textbf{Output:} stream $(a_1,\dots,a_T) \in [d]^T$.
        \Statex \textbf{$\boldsymbol{\adv}$:} \Longunderstack[l]{
         At each timestep $t\in[T]\setminus\{t^*\}$, $\adv$ provides $\Sim$ with record $\dstream_t\in\X$, where $\X = \{0,1\}^d$.
         At \\the challenge timestep~$t^*$ (chosen by $\adv$), it provides records $\dstream_{t^*}^{(L)}, \dstream_{t^*}^{(R)} \in \X$. 
         At each timestep\\ $t\in[T]$, $\Sim$ provides $\adv$ with output $a_t\in [d]$.}
        \Statex \textbf{$\boldsymbol{\idealgauss}$:} $\Sim$ exchanges $\log_2 T + 1$ messages with $\idealgauss$.
        \State \textbf{Initialization:} Perform Step~1 (the initialization phase) of \Cref{alg:selbintree}.
        \State $j \gets 1$.
        \For{\text{$t\in [T]$}}
            \If{\text{$t = t^*$}}
                \State Get input $(\dstream_{t^*}^{(L)}, \dstream_{t^*}^{(R)})$ from $\adv$.
                \For{\text{$i \in [\log T + 1]$}}
                    \State Send $(\dstream_{t^*}^{(L)}, \dstream_{t^*}^{(R)})$ to $\idealgauss$, and get back a response $p_i$.
                \EndFor
            \Else
                \State Get record $\dstream_t$ from $\adv$.
            \EndIf
            \For{\text{each node $\node_{[\nleft,t]}$}}
                \If{$t^* \not\in [\nleft:t]$ where $[\nleft:t]$ denotes the integers $\{\nleft, ..., t\}$}
                    \State Draw noise $Z \sim \gauss(0,\sigma^2 \mathbb{I}^{d \times d})$, where $\sigma = \sqrt{\frac{d (\log T+1)}{2\rho}}$.
                    \State $\nodeval_{[\nleft:t]} \gets Z + \sum_{i=\nleft}^{\nright} \dstream_i$
                \Else\label{algline:contains_challenge_bintree}
                    \State $\node_{[\nleft:t]} \gets \sum_{i \in [\nleft:t]\setminus\{t^*\}} \dstream_i + p_j$. 
                    \State $j \gets j+1$.
                \EndIf
            \EndFor
            \State \Longunderstack[l]{\emph{Output Steps:}} Perform Steps~\ref{step:output-start}--\ref{step:a_tchoice} (the output steps) of \Cref{alg:selbintree}.
        \EndFor
    \end{algorithmic}
\ifnum\pods=0
\end{algorithm}
\else
\end{algorithm*}
\fi

\begin{algorithm}[h!]
\caption{Mechanism $\idealgauss$}
\label{alg:ideal_gauss}
    \begin{algorithmic}[1]
        \Statex \textbf{Input:} $\side \in \{L, R\}$ (not known to $\Sim$).
        \Statex \textbf{Output:} A natural number.
        \For{\text{$i=1$ to $\log T + 1$}}
            \State Get records $\node_i^{(L)}$, $\node_i^{(R)} \in \zo^d$ from $\Sim$.
            \State Draw noise from a multivariate Gaussian distribution $Z \sim \gauss(0,\sigma^2 \mathbb{I}^{d \times d})$, where $\sigma = \sqrt{\frac{d(\log T+1)}{2\rho}}$.\label{algline:idealgauss-sigma}
            \State Output $\node_i^{(\side)} + Z$
        \EndFor
    \end{algorithmic}
\end{algorithm}

\begin{proof}
Consider an adversary $\adv$ interacting with the privacy game $\Pi_{\mech,\adv}$. We want to argue that the adversary's view is $\rho$-close in the two versions of the privacy game (for the two possible values of $\side \in \{L,R\}$.)  We will achieve this by introducing a $\rho$-zCDP mechanism $\idealgauss$ with input $\side$ and reducing our goal to the privacy of $\idealgauss$. 

For this, we use a simulation argument similar to those used in cryptography. Specifically, our proof defines two algorithms: (a) a $\rho$-zCDP mechanism $\idealgauss$ that gets input $\side \in \{L,R\}$ and (b) a simulator $\Sim$ with query access to $\idealgauss$ that does not know the value of $\side$. The simulator $\Sim$ interacts with adversary $\adv$ and satisfies a key guarantee: 
\begin{quote}
    The view of the adversary $\adv$ in its interaction with $\Sim$ is identically distributed to its view in the privacy game $\Pi_{\mech,\adv}$, defined in \Cref{alg:privacy_game}.   (\Cref{fig:sim_diagram} illustrates the structure of these two kinds of interaction.)
\end{quote} 

Since the simulator's outputs to $\adv$ are a post-processing of the query responses from $\idealgauss$, we can argue that the adversary's view is $\rho$-close in the two versions of the privacy game $\Pi_{\mech,\adv}$.

\begin{figure*}
    \centering
    \ifnum\pods=1
    \includegraphics[width={0.9\textwidth}]{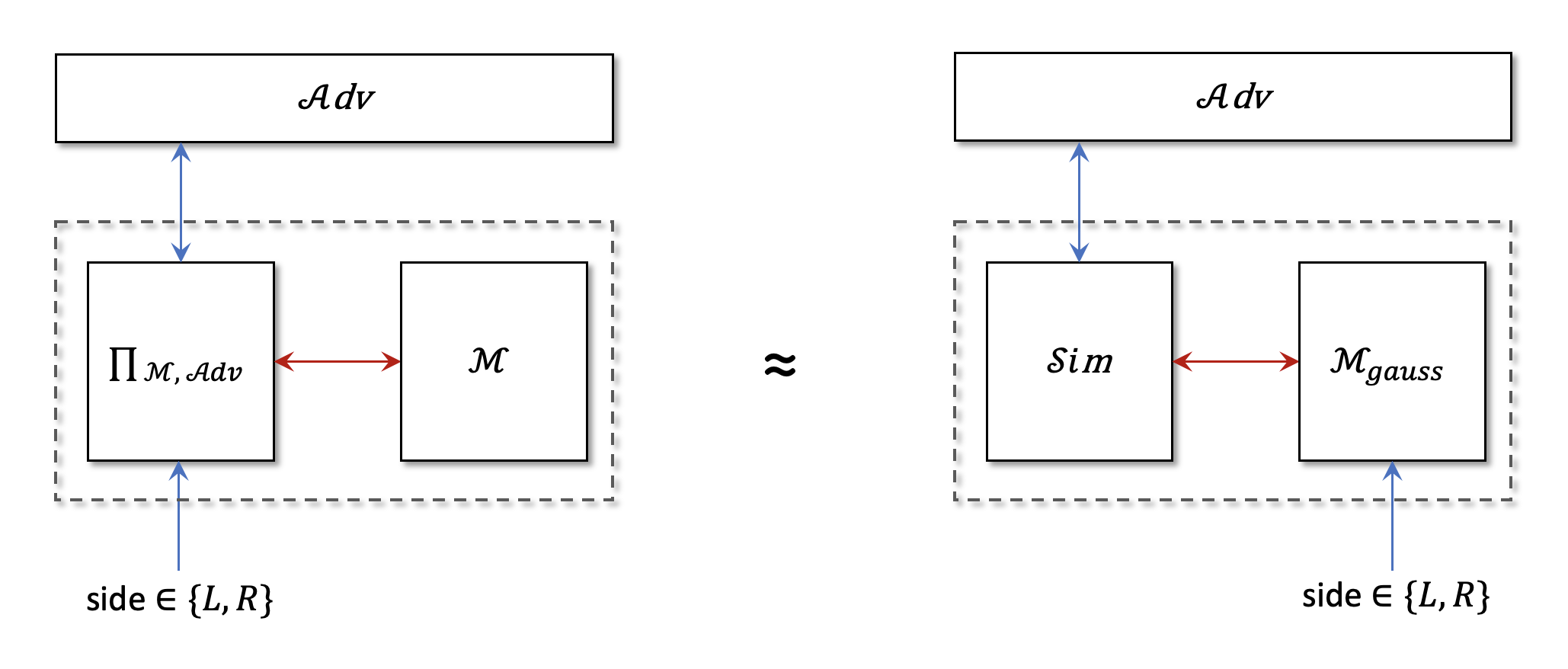}
    \else 
    \includegraphics[scale =0.35]{images/sim_diagram.png}
    \fi
    \caption{An illustration of the simulation argument from the proof of \Cref{lem:privbintree}. The left-hand side shows the game used to define privacy with adaptively selected inputs. The right-hand side shows the simulation structure described in the proof. For each value of $\side$, the adversary's view is identical in these two settings.}
    \label{fig:sim_diagram}
\end{figure*}


To see why this is helpful, recall that we want to show that the probability of $\adv$ guessing the value of $\side$ in the privacy game is small. If the probability of $\adv$ guessing the value of $\side$ is the same in the privacy game as in its interaction with $\Sim$, then—since the simulator doesn't know the value of $\side$—$\adv$ can only learn as much about $\side$ from its interaction with $\Sim$ as one can learn by querying $\idealgauss$. Intuitively, if $\idealgauss$ does not reveal much about the value of $\side$ then neither does $\mech$. We now describe $\idealgauss$ and the simulator, and formalize the argument.

The mechanism $\idealgauss$ (described in \Cref{alg:ideal_gauss}) gets  an input $\side \in \{L,R\}$. It receives at most $\log T + 1$ queries of the form $v^{(L)},v^{(R)}$ from $\Sim$ to which it responds with $p = v^{(\side)} + Z$ where the noise $Z$ is drawn from $\gauss(0,\sigma^2 \mathbb{I}^{d \times d})$ for $\sigma = \sqrt{\frac{d (\log T+1)}{2 \rho}}$. Observe that if $\idealgauss$ has only a single interaction with $\Sim$ and outputs a single noised value, then by the privacy guarantee of the Gaussian mechanism (Lemma~\ref{prop:gaussian-mech}), $\idealgauss$ is $\frac{\rho}{\log T+1}$-zCDP. This can be seen by imagining that $\idealgauss$ is computing a function $f(\side) = \dstream_i^{\side}$ and observing that the $\ell_2$-sensitivity of $f$ is $\sqrt{d}$. Since there are $\log T +1$ interactions between $\idealgauss$ and $\Sim$, $\idealgauss$ is an adaptive composition of $\log T+1$ algorithms, each of which is $\frac{\rho}{\log T+1}$-zCDP. By \Cref{lem:cdp_composition} on composition, $\idealgauss$ is $\rho$-zCDP. 

The simulator $\Sim$ (described in \Cref{alg:binary_sim}) interacts with the adversary without knowing the input $\side \in \{L,R\}$ that is given to $\idealgauss$. It queries $\idealgauss$ exactly $\log T + 1$ times and uses the query responses to provide outputs to the adversary. The aim of the simulator is to mimic the behaviour of $\Pi_{\mech,\adv}$ even though it doesn't know $\side$. The simulator constructs a binary tree as described in \Cref{alg:selbintree}. 
For all nodes in the binary tree except for those whose interval contains the challenge timestep $t^*$, 
the computation of the noisy subtree sums 
can be done by $\Sim$ without any help from $\idealgauss$. For the nodes whose interval does contain $t^*$, the simulator sends $(x_{t^*}^{(L)},x_{t^*}^{(R)})$ to $\idealgauss$ and gets a noisy value of $x_{t^*}^{\side}$. It can then compute the corresponding subtree sum by adding the input records corresponding to the remaining leaves. Notice that the simulator can produce these outputs  online---at the same time that $\Pi_{\mech,\adv}$ would.

The crucial point to note is that the view of the adversary $\adv$ in the privacy game $\Pi_{\mech,\adv}$ is identically distributed to its view in the interaction with $\idealgauss$ and $\Sim$. Furthermore, the view of the adversary $\adv$ when interacting with $\Sim$ and $\idealgauss$ is simply a post-processing of the outputs provided to it by $\Sim$, which are a post-processing of the outputs provided to $\Sim$ by $\idealgauss$. Hence,
$$\idealgauss\text{ is }\rho\text{-zCDP } \implies V_{\mech,\adv}^{(L)}\simeq_{\rho} V_{\mech,\adv}^{(R)}.$$

It remains to argue that exactly $\log T+1$ nodes have a subtree sum that depends on the inputs from the challenge timestep $t^*$. Each node $\node_{[\nleft,\nright]}$ whose subtree sum depends on the inputs from timestep $t^*$ satisfies $t^*\in [\nleft:\nright]$. This holds only for one node at each level of the binary tree created by $\Sim$ (because the intervals represented by the nodes at a particular level are disjoint.) Since the binary tree has depth $\log T + 1$, exactly $\log T+1$ nodes have a subtree sum that depends on the inputs from the challenge timestep $t^*$. 
\end{proof}

\begin{lemma}\label{lem:accbintree}
For all $\rho > 0$ and sufficiently large $T \in \N$, mechanism $\mech$ is $(\alpha,T)$-accurate for $\selection{}$ in the \aCR{}  for $\alpha = \BigO{ \frac{\sqrt{d} \log T \sqrt{\log(Td)}}{\sqrt{\rho}}}$.
\end{lemma}
\begin{proof}
Consider any adversarial process $\adv$ interacting with $\mech$.
We  first argue that, at every timestep~$t$, the random variable $ \err_{\selection{d}}(\xt, a_t)$ corresponding to the error at any timestep~$t$ can be upper bounded by a random variable that is the sum of at most $2\log t$ independent Gaussian random variables. We then use tail bounds for Gaussian random variables, along with a union bound, to argue that, with high probability, the maximum value of this random variable is not too large. Finally, we take a union bound over timesteps to argue that, with high probability, $\max_{t\in [T]} \err_{\selection{d}}(\xt, a_t)$ is not too large.

First, at any timestep~$t$, let $sum_t$ represent the vector of noisy sums defined in Step~\ref{step:noisysumvec} of \Cref{alg:selbintree}. Therefore, each coordinate of this sum, $sum_t[j] = \sum_{i\in [t]} \dstream_t[j] + \sum_{i\in [|I_t|]} Z_i$ is the sum of at most $\log t+1$ noisy interval sums. Here, $Z_i$ is a Gaussian random variable with mean $0$ and standard deviation $\sigma = \sqrt{\frac{d(\log T+1)}{2\rho}}$, and all $Z_i$s are mutually independent. Hence, by the linearity of expectation, and by the linearity of the variance of independent random variables, we get that $\sum_{i\in [|I_t|]} Z_i$ is a Gaussian random variable with mean $0$ and standard deviation $\sqrt{\frac{d(\log T+1)|I_t|}{2\rho}} \leq \sqrt{\frac{d(\log T+1)(\log t + 1)}{2\rho}}$. Consider the vector $N$ consisting of the absolute values of $d$ random variables independently drawn from the distribution $\gauss(0,\sigma^2)$ where $\sigma = \sqrt{\frac{d|I_t|(\log T+1)}{2\rho}}$. The distribution of $N$ is identical to the component-wise absolute values of the Gaussian noise vector $sum_t - \sum_{i\in [t]} \dstream_i$. Then,
$$\sum_{i\in [t]} \dstream_i[a_t] + \max_{j \in [d]} N[j] \geq \maxsum{d}(\xt) - \max_{j \in [d]} N[j],$$ 
since if $a_t$ is selected at timestep~$t$, the noisy sum of coordinate $a_t$ at timestep~$t$ is larger than the noisy sums of all other coordinates at timestep~$t$ (see Step~\ref{step:a_tchoice} in \Cref{alg:selbintree}). Thus,
\ifnum\pods=0
$$\err_{\selection{d}}(\xt, a_t) = \maxsum{d}(\xt) - \sum_{i\in [t]} \dstream_i[a_t] \leq 2 \max_{j \in [d]} N[j].$$
\else
\begin{align*}
&\err_{\selection{d}}(\xt, a_t)\\ = &\quad \maxsum{d}(\xt) - \sum_{i\in [t]} \dstream_i[a_t] \leq 2 \max_{j \in [d]} N[j].
\end{align*}
\fi
Next, we reason about $\max_{j \in [d]} N[j]$ using standard probability tools. Set $\ell = \sqrt{10\frac{d(\log T+1)(\log t+1) \log(dT)}{2\rho}}.$ By Lemma~\ref{lem:gauss_max} on concentration of the maximum of the absolute values of Gaussian random variables, and since
$\sigma \leq \sqrt{\frac{d(\log T+1)(\log t + 1)}{2\rho}}$, we get that
$$\Pr [\max_{j \in [d]} N[j] > \ell] \leq 2d e^{-\frac{\ell^2}{2\sigma^2}} \leq 2d e^{-5\log(dT)} \leq \frac{2}{T^{5}}.$$ 
Then,  with probability at most $\frac{2}{T^{5}}$ (over the coins of the algorithm $\alg$ and the adversarial process $\adv$), 
\begin{align*}
\err(\xt, a_t) &> 20 \sqrt{\frac{d(\log T+1)(\log t+1) \log(dT)}{2\rho}}\\
&\geq 20 \sqrt{\frac{d(\log T+1)^2\log(dT)}{2\rho}},
\end{align*}
since $t \leq T$.
By a union bound over all $t \in [T]$, we get that $\max_{t\in [T]} \err_{\selection{d}}(\xt, a_t) > 20 \sqrt{\frac{d(\log T+1)^2 \log(dT)}{2\rho}}$ 
with probability at most $\frac{2}{T^{4}} \leq \frac{1}{3}$ for sufficiently large~$T$. 
This proves the lemma.
\end{proof}
\begin{proof}[Proof Sketch of Theorem~\ref{thm:bintree_pure}]
The proof of Theorem~\ref{thm:bintree_pure} for $\selection{d}$ closely follows the exposition above. The mechanism used is the same as Algorithm~\ref{alg:selbintree}, except that in Line~\ref{line:noisebintree}, $Z$ is drawn from $Lap(\frac{d(\log T + 1)}{\eps})$ instead of a Gaussian distribution. The privacy proof is exactly as in Lemma~\ref{lem:privbintree}, except that we use that the composition of $\log T + 1$ mechanisms that are $\big(\frac{\eps}{\log T + 1},0\big)$-DP is $(\eps, 0)$-DP instead a composition theorem for $\rho$-zCDP. The accuracy proof closely follows that of Lemma~\ref{lem:accbintree}, with the main difference being that the the vector $N$ is defined as the component-wise absolute value of $d$ random variables independently drawn from the distribution of the sum of $|I_t|$ independent random variables distributed as $Lap\left(\frac{d(\log(T)+1)}{\eps}\right)$. We  then use the concentration inequality for the maximum of the absolute values of independent Laplace random variables over $d|I_t|$ random variables in Lemma~\ref{lem:lapmax} with $a = 2 \log T$ to argue that the absolute value of each Laplace random variable is smaller than 
$\frac{d(\log T + 1)}{\eps}(\log(d|I_t|) + 10 \log T)$ with probability at least $\frac{1}{T^{10}}$. This implies that $\max_{j \in [d]} N[j]$ is smaller than 
$\frac{d(\log T + 1)^2}{\eps}(\log(d(\log T+1)) + 10 \log T)$ with probability at least $\frac{1}{T^{10}}$, upper bounding $|I_t|$ by $\log T + 1$. Taking a union bound over $T$ and using the fact that $\log(d(\log T + 1)) \leq \log d (\log T + 1)$ for sufficiently large $T$ completes the proof.
\end{proof}
Theorems~\ref{thm:bintree_rho} and~\ref{thm:bintree_pure} for $\maxsum{d}$ are proved analogously. The main difference is that we output $\max_{j \in [d]} sum_t[j]$ instead of $\argmax_{j \in [d]} sum_t[j]$ in \Cref{line:outputselect} of \Cref{alg:selbintree}.

\subsection{Algorithms that Recompute at Regular Intervals}\label{sec:maxsum-upperbounds}

In this section, we prove Item~1 of \Cref{thm:recompute_rho} for sensitivity-1 functions. The proof of Item~2 of \Cref{thm:recompute_rho} builds on  the same idea of recomputing $\selection{d}$ every $T/m$ timesteps, but it uses the report noisy max (with exponential noise) algorithm for $\selection{d}$ \cite{McSheldon} instead of adding Gaussian noise to the function. We omit the details, since the argument is essentially the same as in the rest of this section.

The mechanism recomputes the function every $r$ timesteps. Between recomputations, it outputs the most recently computed value. We select $r$ to balance the privacy cost of composition with the error due to returning stale values between recomputations.

\begin{algorithm}[H]
    \caption{Mechanism $\mech$ for sensitivity-$1$ functions in \aCR}
    \label{alg:maxsum}
    \hspace*{\algorithmicindent}
    \begin{algorithmic}[1] 
        \Statex\textbf{Input:} time horizon $T$, privacy parameter $\rho > 0$, recompute period $r \in [T-1]$, function $f$, stream $\vec{\dstream} = (\dstream_1,\dots,\dstream_T) \in \X^n$ where $\X = \{0,1\}^d$.  \Statex\textbf{Output:} stream $(a_1,\dots,a_T) \in \R^T$.
        \State $m \gets \lfloor\frac{T-1}{r}\rfloor$.
        \For{\text{$k=1$ to $m$}}
        \State Get input record $\dstream_{(k-1)r+1}$. 
        \State Draw $Z_k \sim \gauss(0,\sigma^2)$, where $\sigma = \sqrt{\frac{m}{2\rho}}$.   
               \label{line:gaussnoiserecom}
               \State Output $a_{(k-1)r+1} \gets f(\vec{\dstream}_{[(k-1)r+1]} ) + Z_k$. \label{line:noise_sensone}
               \For{\text{$t=(k-1)r+2$ to $kr$}}
                  \State Get input record $\dstream_t$.
                  \State Output $a_t \gets a_{(k-1)r+1}$. \label{line:recompute_sensone}
               \EndFor
             \EndFor
    \end{algorithmic}
\end{algorithm}

\begin{claim}\label{lem:privmax}
For all $\rho, T > 0$, $r \in [T-1]$, mechanism $\mech$ defined in \Cref{alg:maxsum} is $\rho$-zCDP in the \aCR{}.  
\end{claim}

\begin{proof}
Consider an adversary $\adv$ interacting with $\mech$. We define a mechanism \sstext{$\idealrecomp$}, similar to \Cref{alg:ideal_gauss}, and a simulator $\Sim$ that interacts with the adversary $\adv$ such that the view of adversary $\adv$ in the interaction with $\idealrecomp$ and $\Sim$ is identically distributed to its view in the privacy game $\Pi_{\mech,\adv}$, defined in \Cref{alg:privacy_game}. 

\begin{algorithm}[h!]
\caption{Mechanism $\idealrecomp$}
\label{alg:ideal_recomp}
    
    \begin{algorithmic}[1]
        \Statex \textbf{Input:} $\side \in \{L, R\}$ (not known to $\Sim$) 
        \Statex \textbf{Output:} A natural number.
        \State Get neighboring datasets $\vec{\dset}^{(L)}$, $\vec{\dset}^{(R)} \in \zo^d$ and a function $f$ with $\ell_2$ sensitivity  at most $1$ from $\Sim$.
        \State Draw noise $Z \sim \gauss(0,\sigma^2)$, where $\sigma = \sqrt{\frac{m}{2\rho}}$.
        \State Output $f\left(\dset^{(\side)}\right) + Z$
    \end{algorithmic}
\end{algorithm}

\ifnum\pods = 1
\begin{algorithm*}[h!]
\else
\begin{algorithm}[h!]
\fi
    \caption{Simulator $\Sim$ for the proof of \Cref{lem:privmax}}
    \label{alg:recompute_sim}
    \begin{algorithmic}[1]
        \Statex \textbf{Input:} time horizon $T$, privacy parameter $\rho > 0$, recompute period $r \in [T-1]$, function $f$. $\Sim$ also has black-box access to an adversary $\adv$ and a process $\idealrecomp$.
        \Statex \textbf{Output:} stream $(a_1,\dots,a_T) \in \R^T$
        \Statex \textbf{$\boldsymbol{\adv}$:} \Longunderstack[l]{
         At each timestep $t\in[T]\setminus\{t^*\}$, $\adv$ provides $\Sim$ with record $\dstream_t\in\X$. At the challenge timestep~$t^*$\\ (chosen by $\adv$), it provides two records $\dstream_{t^*}^{(L)}, \dstream_{t^*}^{(R)}\in\X$. At every timestep $t\in[T]$, $\Sim$ provides \\$\adv$ with output $a_t\in \R$.}
        \Statex \textbf{$\boldsymbol{\idealrecomp}$:} $\Sim$ exchanges $T/r$ messages with $\idealrecomp$.
        \State \textbf{Initialization:} $m \gets \as{\lceil\frac{T}{r}\rceil} $, $j \gets 1$.
        \State For timesteps $t<t^*$, run mechanism $\mech$ in \Cref{alg:maxsum} with inputs from $\adv$, with the same $T, \rho, r, f$. Let $a_t$ be $\mech$'s output at timestep $t$.
        \For{\text{$t \geq t^*$}}
            \If{\text{$t = t^*$}}
                \State Get input $(\dstream_{t^*}^{(L)}, \dstream_{t^*}^{(R)})$ from $\adv$. 
            \EndIf
            \If{\text{$t \neq t^*$}}
                \State Get record $\dstream_t$ from $\adv$.
            \EndIf
            \If{$t \mod r = 1$}
                \State Let $\vec{\dset}^{(\side)}_t = \{\dstream_1,\dots,\dstream_{t^*-1},\dstream^{(\side)}_{t^*},\dstream_{t^*+1},\dots,\dstream_t\}$ \as{for each $\side \in \{L,R\}$}
                \State \sstext{ 
                $a_t \gets \idealrecomp\left(f, \vec{\dset}_t^{(L)}, \vec{\dset}_t^{(R)}\right)$.}
            \Else
                \State $q \gets \lfloor \frac{t}{r}\rfloor$; output $a_t \gets a_{q +1}$. 
            \EndIf
            \EndFor
    \end{algorithmic}
\ifnum\pods = 0
\end{algorithm}
\else
\end{algorithm*}
\fi


\sstext{The mechanism $\idealrecomp$ is defined in Algorithm~\ref{alg:ideal_recomp}. Since the function $f$ has $\ell_2$ sensitivity at most $1$, then by the privacy of the Gaussian mechanism, and since the variance of the noise added is $\frac{m}{2\rho}$), $\idealrecomp$ is $\frac{\rho}{m}$-zCDP with respect to the dataset consisting of $\side \in \{L,R\}$.} 

The simulator $\Sim$ (described in \Cref{alg:recompute_sim}) gets inputs from $\adv$, but it does not know the input $\side \in \{L,R\}$ that is given to \sstext{$\idealrecomp$}. It interacts with \sstext{$\idealrecomp$} to provide outputs to the adversary $\adv$. The aim of the Simulator is to mimic the behaviour of $\Pi_{\mech,\adv}$ even though it doesn't know $\side$. For all timesteps $t<t^*$ before the challenge timestep, the simulator behaves exactly like $\mech$. 
%
%
\sstext{Starting at the challenge timestep, for every $t\in[t^*:T]$ where $\mech$ would recompute the noised value of the sum, $\Sim$ sends $\idealrecomp$ the function $f$ as well as neighboring datasets  $\vec{\dset}_t^{(L)}, \vec{\dset}_t^{(R)}$ defined by
$$\vec{\dset}^{(\side)}_t = \{\dstream_1,\dots,\dstream_{t^*-1},\dstream^{(\side)}_{t^*},\dstream_{t^*+1},\dots,\dstream_t\}.$$
Since $\Sim$ queries $\idealrecomp$ at most $m$ times, by adaptive composition, the output transcript of $\idealrecomp$ is $\rho$-zCDP with respect to the dataset consisting of $\side$.}

The view of the adversary $\adv$ in the real privacy game $\Pi_{\mech,\adv}$ is identically distributed to its view in the interaction with \sstext{$\idealrecomp$} and $\Sim$. Furthermore, the view of the adversary $\adv$ when interacting with $\Sim$ and \sstext{$\idealrecomp$} is simply a post-processing of the outputs provided to it by $\Sim$, which are a post-processing of the outputs provided to $\Sim$ by \sstext{$\idealrecomp$}. 
\sstext{As argued previously, the output  of $\idealrecomp$} when $\side = L$ is $\rho$-close to its output \sstext{transcript} when $\side  = R$. Hence we have that $V_{\mech,\adv}^{(L)}\simeq_{\rho} V_{\mech,\adv}^{(R)}$.
\end{proof}

\begin{claim}\label{lem:acc-max-delta}
Fix $\rho > 0$, sufficiently large $T > 0$, and $2 \leq m \leq T$. Let $f:\X^* \to \Z$ be a function with $\ell_2$-sensitivity at most $1$. Then mechanism $\mech$, defined in \Cref{alg:maxsum} is $(\alpha, T)$-accurate for $f$ in the \aCR{} where $\alpha = \frac{T}{m} + \sqrt{\frac{10m \log m}{\rho}}.$
\end{claim}

\begin{proof}
Consider any adversarial process $\adv$ interacting with $\mech$. Fix a timestep~$t \in [T]$. Consider time horizon $T$ divided into $m$ stages, where the stage $k\in[m]$ is from timestep $(k-1)r+ 1$ to $kr$. Let timestep~$t$ be in  stage $k$. Intuitively, since $\mech$, defined in \Cref{alg:maxsum}, corresponds to recomputing the noisy sum every $r$ timesteps (and using each recomputed value for the next $r$ timesteps), the error can be decomposed into two parts: one caused by the drift in the true value of the function since the last recomputation and the other caused by noise addition. By the triangle inequality,
\ifnum\pods=0
\begin{align*}
    \err_f(\xt, a_t) 
    &= \abs{a_t - f(\xt)} \\
    &\leq \abs{f(\xt) - f(\vec{\dstream}_{[(k-1)r+1}]}) + \abs{a_t - f(\vec{\dstream}_{[(k-1)r+1}]}) \\
    & \leq T/m + \abs{a_{(k-1)r+1} - f(\vec{\dstream}_{[(k-1)r+1]})} 
    \leq \frac{T}{m} + \abs{Z_k}.
\end{align*}
\else
\begin{align*}
    & \err_f(\xt, a_t) 
    = \abs{a_t - f(\xt)} \\
    &\leq \abs{f(\xt) - f(\vec{\dstream}_{[(k-1)r+1}]}) + \abs{a_t - f(\vec{\dstream}_{[(k-1)r+1}]}) \\
    & \leq T/m + \abs{a_{(k-1)r+1} - f(\vec{\dstream}_{[(k-1)r+1]})} 
    \leq \frac{T}{m} + \abs{Z_k}.
\end{align*}
\fi
The second inequality above holds because the $\ell_2$-sensitivity of $f$ is at most $1$, and since we recompute every $r=T/m$ timesteps, the maximum change in the function $f$ since the last recomputation is $T/m$. The third inequality follows from Steps~\ref{line:noise_sensone} and~\ref{line:recompute_sensone} in \Cref{alg:maxsum}. Finally, observe that $Z_k$ for $k \in [m]$ are mutually independent Gaussian random variables with mean $0$ and standard deviation $\sqrt{\frac{m}{2\rho}}$. Hence, applying Lemma~\ref{lem:gauss_max} on the concentration of the maximum of the absolute values of Gaussian random variables (setting $\ell = \sqrt{\frac{10 m \log m}{\rho}}$), and using the fact that $m \geq 2$, 
\ifnum\pods=0
\begin{align*}
    \Pr_{\text{coins of }\alg,\adv}\left( \max_{t\in[T]} \err_f(\xt,a_t) \geq \frac{T}{m} + \sqrt{ \frac{10m \log m}{\rho}} \right) & =
    \Pr_{\text{coins of }\alg,\adv}\left( \max_{k\in[m]} |Z_k| \geq \sqrt{ \frac{10m \log m}{\rho}} \right) \\
    & \leq \frac{2}{m^9} \leq \frac{1}{3}.\qedhere
\end{align*}
\else
\begin{align*}
    & \Pr_{\text{coins of }\alg,\adv}\left( \max_{t\in[T]} \err_f(\xt,a_t) \geq \frac{T}{m} + \sqrt{ \frac{10m \log m}{\rho}} \right) \\ =
    & \Pr_{\text{coins of }\alg,\adv}\left( \max_{k\in[m]} |Z_k| \geq \sqrt{ \frac{10m \log m}{\rho}} \right) \leq \frac{2}{m^9} \leq \frac{1}{3}. \qedhere
\end{align*}
\fi
\end{proof}
\begin{proof}[Proof of Item~1 in \Cref{thm:recompute_rho}]
By Claim~\ref{lem:privmax}, the mechanism $\mech$ is $\rho$-zCDP in the \aCR{}. 

For $\rho \leq \frac{\log T}{T^2}$, consider the mechanism that doesn't touch the data and always outputs $0$. Clearly it is $0$-zCDP. Additionally, for this mechanism, $\alpha = O(T)$. For $\rho > \frac{\log T}{T^2}$, by Claim~\ref{lem:acc-max-delta}, mechanism $\mech$ is $(\alpha, T)$-accurate for $f$ in the \aCR{}, where $\alpha = T/m + 10\sqrt{\frac{m \log m}{\rho}}$. Setting $m = \lfloor \frac{\rho^{1/3} T^{2/3}}{\log^{1/3} T} \rfloor$ gives $\alpha = \BigO{\min \left\{T, \sqrt[3]{\frac{T \log T}{\rho}} \right\}}$, where the $\min$ comes from the option of using the trivial mechanism. 
\end{proof}

\begin{proof}[Proof Sketch of Item~1 in \Cref{thm:recompute_pure}]
The mechanism $\mech$ used is a variant of \Cref{alg:maxsum}. The only difference is that in Line~\ref{line:gaussnoiserecom}, instead of the random variable $Z_k$ being distributed as a Gaussian, it is distributed as $Lap(\frac{m}{\eps})$. The privacy proof follows a structure similar to that of Claim~\ref{lem:privmax}, with the main difference being that instead of using a composition theorem for $\rho$-zCDP, we instead use that the composition of $m$ mechanisms that are $(\frac{\eps}{m},0)$-DP is $(\eps, 0)$-DP. 

For accuracy, we can prove a claim phrased exactly as  Claim~\ref{lem:acc-max-delta}, with $\alpha = \frac{T}{m} + \frac{m}{\eps}[ \log m + 2 \log T ]$ instead of $\alpha = \frac{T}{m} + \sqrt{\frac{10m \log m}{\rho}}$. The proof is similar, with the only difference being that instead of using Lemma~\ref{lem:gauss_max} on the maximum of i.i.d. Gaussian random variables, we instead use Lemma~\ref{lem:lapmax} on the maximum of i.i.d. Laplace random variables, with $t = 2\log T$.

Finally, we prove the theorem as follows: for $\eps > \frac{\log T}{T}$, setting $m = \lfloor \sqrt{\frac{\eps T}{\log T}} \rfloor$ in the accuracy claim gives $\alpha = O(\sqrt{\frac{T}{\eps}\log T})$. For $\eps \leq \frac{\log T}{T}$, we can consider the mechanism that always outputs $0$ at every timestep. This mechanism is $(0,0)$-DP and $(\alpha, T)$-accurate for $f$ in the \aCR{} with $\alpha = O(T)$. This completes the proof.
\end{proof}
\begin{proof}[Proof Sketch of Item~2 in Theorems~\ref{thm:recompute_rho} and \ref{thm:recompute_pure}]
We \ifnum\pods=1 \\ \fi  sketch the proof of Item 2 of \Cref{thm:recompute_rho}. The proof of Item 2 of \Cref{thm:recompute_pure} is essentially the same. The upper bound mechanism $\mech$ used for this proof is a variant of \Cref{alg:maxsum} where we recompute $\selection{d}$ using the exponential mechanism \cite{McTalwar} with $\eps' = \sqrt{\frac{2\rho}{m}}$ (for Item 2 of \Cref{thm:recompute_pure} on pure DP, we use $\eps' = \frac{\eps}{m}$). The quality function of an attribute and dataset pair is defined to be the sum of that attribute over all entries in the dataset. The exponential mechanism instantiated as described above is used to privately compute $\selection{d}$ every $T/m$ timesteps. Between recomputations, the attribute index produced at the last recomputation is used as the output. 

The privacy proof follows a structure similar to that of Claim~\ref{lem:privmax}. The main difference for this proof is that the simulator will now interact with an ideal mechanism that takes as input a differentially private algorithm as well as neighboring datasets to run the algorithm on. In particular, the neighboring datasets will be the inputs $x_{t^*}^{(L)}$ and $x_{t^*}^{(R)}$ from the challenge timestep, and the algorithm will be the exponential mechanism hardcoded with all the inputs of the adversary so far (except for the inputs from the challenge timestep.) The ideal mechanism will run the algorithm with challenge input $x_{t^*}^{(\side)}$ and output the result. The adversary's view in the privacy game is clearly identical to its view when interacting with the simulator. Finally, the closeness of the adversary's view in the simulated world when $\side = L$ and when $\side = R$ follows directly from the privacy of the exponential mechanism and adaptive composition~\cite{DworkRV10,BunS16}.

For accuracy, we prove a claim akin to  Claim~\ref{lem:acc-max-delta}, with $\alpha = \frac{T}{m} + 2\sqrt{\frac{m}{2\rho}}[ \log d + 5 \log m ]$. The proof is similar to that of Claim~\ref{lem:acc-max-delta}; here, we define $|Z_k|$ as the error incurred by the $k^{th}$ instantiation of the exponential mechanism, and use Lemma~\ref{lem:expmech} on the accuracy of the exponential mechanism (setting $a = 5 \log m$) and take a union bound over the $m$ recomputations to argue that the maximum error is greater than $\alpha = \frac{T}{m} + 2\sqrt{\frac{m}{2\rho}}[ \log d + 5 \log m ]$ with probability at most $\frac{1}{m^4}$.

For $\rho > (\frac{\log (d T)}{T})^2$, by the accuracy claim, mechanism $\mech$ is $(\alpha, T)$-accurate for $f$ in the \aCR{}, where $\alpha = \frac{T}{m} + 2\sqrt{\frac{m}{2\rho}}[ \log d + 2 \log m ]$. Setting $m = \lfloor \frac{\rho^{1/3} T^{2/3}}{(\log (d T)^{2/3}} \rfloor$  yields $\alpha = O\left(\frac{T^{1/3}\log (d T)^{2/3} }{\rho^{1/3}} \right)$. Finally, for $\rho \leq (\frac{\log (d T)}{T})^2$, consider the mechanism that doesn't touch the data and always outputs $0$. It is clearly $0$-zCDP, and has $\alpha = O(T)$.
\end{proof}




\ifnum\pods=0
\section*{Acknowledgments}

We are grateful to Kobbi Nissim for being part of the conversations that got this work started and for subsequent helpful comments. We are also grateful to Jon Ullman for insights into  the difficulty of the top-$k$ selection problem. 
\else 
\fi

\ifnum\pods=0
\addcontentsline{toc}{section}{References}
\bibliographystyle{plain}
\bibliography{bibliography}

\begin{thebibliography}{10}

\bibitem{AgarwalS17}
Naman Agarwal and Karan Singh.
\newblock The price of differential privacy for online learning.
\newblock In Doina Precup and Yee~Whye Teh, editors, {\em Proceedings of the
  34th International Conference on Machine Learning}, volume~70 of {\em
  Proceedings of Machine Learning Research}, pages 32--40. PMLR, 06--11 Aug
  2017.

\bibitem{EliezerJWY20}
Omri Ben-Eliezer, Rajesh Jayaram, David~P. Woodruff, and Eylon Yogev.
\newblock A framework for adversarially robust streaming algorithms.
\newblock In {\em Proceedings of the 39th ACM SIGMOD-SIGACT-SIGAI Symposium on
  Principles of Database Systems}, PODS'20, page 63–80, New York, NY, USA,
  2020. Association for Computing Machinery.

\bibitem{BolotFMNT13}
Jean Bolot, Nadia Fawaz, S.~Muthukrishnan, Aleksandar Nikolov, and Nina Taft.
\newblock Private decayed predicate sums on streams.
\newblock In {\em Proceedings of the 16th International Conference on Database
  Theory}, ICDT '13, page 284–295, New York, NY, USA, 2013. Association for
  Computing Machinery.

\bibitem{BunS16}
Mark Bun and Thomas Steinke.
\newblock Concentrated differential privacy: Simplifications, extensions, and
  lower bounds.
\newblock In Martin Hirt and Adam~D. Smith, editors, {\em Theory of
  Cryptography - 14th International Conference, {TCC} 2016-B, Beijing, China,
  October 31 - November 3, 2016, Proceedings, Part {I}}, volume 9985 of {\em
  Lecture Notes in Computer Science}, pages 635--658, 2016.

\bibitem{bunUV18}
Mark Bun, Jonathan Ullman, and Salil Vadhan.
\newblock Fingerprinting codes and the price of approximate differential
  privacy.
\newblock {\em SIAM Journal on Computing}, 47(5):1888--1938, 2018.

\bibitem{censuscite}
Census Bureau.
\newblock Census disclosure avoidance system, 2020.
\newblock
  \url{https://www.census.gov/programs-surveys/decennial-census/decade/2020/planning-management/process/disclosure-avoidance.html}.

\bibitem{CardosoR21}
Adrian~Rivera Cardoso and Ryan Rogers.
\newblock Differentially private histograms under continual observation:
  Streaming selection into the unknown.
\newblock {\em CoRR}, abs/2103.16787, 2021.

\bibitem{ChanSS10}
T.{-}H.~Hubert Chan, Elaine Shi, and Dawn Song.
\newblock Private and continual release of statistics.
\newblock {\em {IACR} Cryptol. ePrint Arch.}, 2010:76, 2010.

\bibitem{DworkKMMN06}
Cynthia Dwork, Krishnaram Kenthapadi, Frank McSherry, Ilya Mironov, and Moni
  Naor.
\newblock Our data, ourselves: Privacy via distributed noise generation.
\newblock In {\em International Conference on the Theory and Applications of
  Cryptographic Techniques}, EUROCRYPT '06, pages 486--503, St.~Petersburg,
  Russia, 2006.

\bibitem{dwork2006calibrating}
Cynthia Dwork, Frank McSherry, Kobbi Nissim, and Adam Smith.
\newblock Calibrating noise to sensitivity in private data analysis.
\newblock In {\em Theory of cryptography conference}, pages 265--284. Springer,
  2006.

\bibitem{DworkNPR10}
Cynthia Dwork, Moni Naor, Toniann Pitassi, and Guy~N. Rothblum.
\newblock Differential privacy under continual observation.
\newblock In Leonard~J. Schulman, editor, {\em Proceedings of the 42nd {ACM}
  Symposium on Theory of Computing, {STOC} 2010, Cambridge, Massachusetts, USA,
  5-8 June 2010}, pages 715--724. {ACM}, 2010.

\bibitem{DworkNPRY10}
Cynthia Dwork, Moni Naor, Toniann Pitassi, Guy~N. Rothblum, and Sergey
  Yekhanin.
\newblock Pan-private streaming algorithms.
\newblock In Andrew~Chi{-}Chih Yao, editor, {\em Innovations in Computer
  Science - {ICS} 2010, Tsinghua University, Beijing, China, January 5-7, 2010.
  Proceedings}, pages 66--80. Tsinghua University Press, 2010.

\bibitem{DworkRV10}
Cynthia Dwork, Guy~N. Rothblum, and Salil~P. Vadhan.
\newblock Boosting and differential privacy.
\newblock In {\em 51th Annual {IEEE} Symposium on Foundations of Computer
  Science, {FOCS} 2010, October 23-26, 2010, Las Vegas, Nevada, {USA}}, pages
  51--60. {IEEE} Computer Society, 2010.

\bibitem{FichtenHO21}
Hendrik Fichtenberger, Monika Henzinger, and Wolfgang Ost.
\newblock Differentially private algorithms for graphs under continual
  observation.
\newblock In Petra Mutzel, Rasmus Pagh, and Grzegorz Herman, editors, {\em 29th
  Annual European Symposium on Algorithms, {ESA} 2021, September 6-8, 2021,
  Lisbon, Portugal (Virtual Conference)}, volume 204 of {\em LIPIcs}, pages
  42:1--42:16. Schloss Dagstuhl - Leibniz-Zentrum f{\"{u}}r Informatik, 2021.

\bibitem{SmithT13}
Abhradeep Guha~Thakurta and Adam Smith.
\newblock {(Nearly)} optimal algorithms for private online learning in
  full-information and bandit settings.
\newblock In C.~J.~C. Burges, L.~Bottou, M.~Welling, Z.~Ghahramani, and K.~Q.
  Weinberger, editors, {\em Advances in Neural Information Processing Systems},
  volume~26. Curran Associates, Inc., 2013.

\bibitem{HardtLM12}
Moritz Hardt, Katrina Ligett, and Frank McSherry.
\newblock A simple and practical algorithm for differentially private data
  release.
\newblock In Peter~L. Bartlett, Fernando C.~N. Pereira, Christopher J.~C.
  Burges, L{\'{e}}on Bottou, and Kilian~Q. Weinberger, editors, {\em Advances
  in Neural Information Processing Systems 25}, pages 2348--2356, 2012.

\bibitem{HardtT10}
Moritz Hardt and Kunal Talwar.
\newblock On the geometry of differential privacy.
\newblock In {\em Proceedings of the 42nd Annual ACM Symposium on the Theory of
  Computing}, STOC '10, pages 705--714, New York, NY, USA, 2010. ACM.

\bibitem{HassidimKMMS20}
Avinatan Hassidim, Haim Kaplan, Yishay Mansour, Yossi Matias, and Uri Stemmer.
\newblock Adversarially robust streaming algorithms via differential privacy.
\newblock In Hugo Larochelle, Marc'Aurelio Ranzato, Raia Hadsell,
  Maria{-}Florina Balcan, and Hsuan{-}Tien Lin, editors, {\em Advances in
  Neural Information Processing Systems 33: Annual Conference on Neural
  Information Processing Systems 2020, NeurIPS 2020, December 6-12, 2020,
  virtual}, 2020.

\bibitem{jainkt12}
Prateek Jain, Pravesh Kothari, and Abhradeep Thakurta.
\newblock Differentially private online learning.
\newblock In Shie Mannor, Nathan Srebro, and Robert~C. Williamson, editors,
  {\em Proceedings of the 25th Annual Conference on Learning Theory}, volume~23
  of {\em Proceedings of Machine Learning Research}, pages 24.1--24.34,
  Edinburgh, Scotland, 25--27 Jun 2012. JMLR Workshop and Conference
  Proceedings.

\bibitem{KaplanMNS21}
Haim Kaplan, Yishay Mansour, Kobbi Nissim, and Uri Stemmer.
\newblock Separating adaptive streaming from oblivious streaming.
\newblock {\em CoRR}, abs/2101.10836, 2021.

\bibitem{McSheldon}
Ryan McKenna and Daniel~R Sheldon.
\newblock Permute-and-flip: A new mechanism for differentially private
  selection.
\newblock In H.~Larochelle, M.~Ranzato, R.~Hadsell, M.~F. Balcan, and H.~Lin,
  editors, {\em Advances in Neural Information Processing Systems}, volume~33,
  pages 193--203. Curran Associates, Inc., 2020.

\bibitem{McTalwar}
Frank McSherry and Kunal Talwar.
\newblock Mechanism design via differential privacy.
\newblock In {\em Proceedings of the 48th Annual IEEE Symposium on Foundations
  of Computer Science}, FOCS '07, page 94–103, USA, 2007. IEEE Computer
  Society.

\bibitem{PerrierAK19}
Victor Perrier, Hassan~Jameel Asghar, and Dali Kaafar.
\newblock Private continual release of real-valued data streams.
\newblock In {\em 26th Annual Network and Distributed System Security
  Symposium, {NDSS} 2019, San Diego, California, USA, February 24-27, 2019}.
  The Internet Society, 2019.

\bibitem{Renyi61}
Alfred R\'enyi.
\newblock On measures of entropy and information.
\newblock {\em Proceedings of the Fourth Berkeley Symposium on Mathematical
  Statistics and Probability, Volume 1: Contributions to the Theory of
  Statistics, pages 547--561, Berkeley, Calif., 1961. University of California
  Press}, abs/2101.10836, 1961.

\bibitem{SongLMVC18}
Shuang Song, Susan Little, Sanjay Mehta, Staal~A. Vinterbo, and Kamalika
  Chaudhuri.
\newblock Differentially private continual release of graph statistics.
\newblock {\em CoRR}, abs/1809.02575, 2018.

\bibitem{US17}
Thomas Steinke and Jonathan~R. Ullman.
\newblock Tight lower bounds for differentially private selection.
\newblock In Chris Umans, editor, {\em 58th {IEEE} Annual Symposium on
  Foundations of Computer Science, {FOCS} 2017, Berkeley, CA, USA, October
  15-17, 2017}, pages 552--563. {IEEE} Computer Society, 2017.

\bibitem{TalwarTZ15}
Kunal Talwar, Abhradeep Thakurta, and Li~Zhang.
\newblock Nearly optimal private {LASSO}.
\newblock In Corinna Cortes, Neil~D. Lawrence, Daniel~D. Lee, Masashi Sugiyama,
  and Roman Garnett, editors, {\em Advances in Neural Information Processing
  Systems 28: Annual Conference on Neural Information Processing Systems 2015,
  December 7-12, 2015, Montreal, Quebec, Canada}, pages 3025--3033, 2015.

\bibitem{Ullman21pers}
Jonathan Ullman, 2021.
\newblock Personal communication.

\end{thebibliography}


\appendix
\ifnum\pods=0
\newpage
\appendix
\section{Useful Concentration Inequalities}
\begin{lemma}\label{lem:gauss_conc}
For all random variables $R \sim \gauss(0,\sigma^2)$,
\begin{equation*}
    \Pr[|R| > \ell] \leq 2e^{-\frac{\ell^2}{2 \sigma^2}}.
\end{equation*}
\end{lemma}
\begin{lemma}\label{lem:gauss_max}
Consider $m$ random variables $R_1,\dots,R_m \sim \gauss(0,\sigma^2)$. Then
\begin{equation*}
    \Pr[\max_{j \in [m]} |R_j| > \ell] \leq 2 m e^{-\frac{\ell^2}{2 \sigma^2}}.
\end{equation*}
\end{lemma}
\begin{proof}
By a union bound and Lemma~\ref{lem:gauss_conc},
\begin{align*}
    \Pr[\max_{i \in [m]} |R_i| > \ell] 
      = \Pr( \exists i \in [m] \text{ such that } |R_i| > \ell) \\
     \leq \sum_{i=1}^m \Pr( |R_i| > \ell) \leq
    \sum_{i=1}^m  2e^{-\frac{\ell^2}{2 \sigma^2}}  =  2m e^{-\frac{\ell^2}{2 \sigma^2}}.
\end{align*}
\end{proof}
A similar union bound argument yields the following concentration inequality on the maximum of the absolute values of i.i.d.\ Laplace random variables.
\begin{lemma}\label{lem:lapmax}
Fix $m \in \N$, $\lambda > 0$. Consider $m$ random variables $R_1,\dots,R_m \sim Lap(\lambda)$. Then for all $a>0$,
$$\Pr(\max_{i \in [m]} |R_i| > \lambda (\log m + \log a) ) \leq e^{-a}.$$
\end{lemma} 

\fi
\else
\fi

\end{document}